\documentclass[12pt]{article}%
\usepackage{amsmath,amsthm,amssymb,amscd}
\usepackage{latexsym}
\usepackage{amsfonts}
\usepackage{times,epsfig}
\usepackage{natbib}

\newcommand{\Date}[1]{\def\@Date{#1}}
\def\today{\number\day~\ifcase\month\or
 January\or February\or March\or April\or May\or June\or
 July\or August\or September\or October\or November\or December\fi~\number\year}

\newtheorem{theorem}{Theorem}
\newtheorem{lemma}{Lemma}

\def\T{{ \mathrm{\scriptscriptstyle T} }}

\newcommand{\be}{\begin{equation}}
\newcommand{\ee}{\end{equation}}
\newcommand{\bs}{\begin{eqnarray*}}
\newcommand{\es}{\end{eqnarray*}}

\renewcommand{\hat}{\widehat}
\renewcommand{\tilde}{\widetilde}
\newcommand{\wh}{\widehat}
\newcommand{\wt}{\widetilde}

\newcommand{\bA}{{A}}
\newcommand{\bB}{{B}}

\newcommand{\bF}{{F}}
\newcommand{\bH}{{H}}
\newcommand{\bI}{{I}}
\newcommand{\bJ}{{J}}

\newcommand{\bS}{{S}}

\newcommand{\bX}{{X}}

\newcommand{\bb}{\textit{b}}

\newcommand{\bv}{\textit{v}}
\newcommand{\bx}{\textit{x}}
\newcommand{\by}{\textit{y}}

\newcommand{\bzero}{\textit{0}}
\newcommand{\bbe}{{e}}
\newcommand{\ep}{\varepsilon}
\newcommand{\bep}{\varepsilon}
\newcommand{\bSigma}{{\Sigma}}

\newcommand{\bbeta}{{\beta}}
\newcommand{\bPhi}{{\Phi}}

\newcommand{\bic}{\textsc{BIC}}
\newcommand{\cov}{\text{cov}}
\newcommand{\rss}{\textsc{RSS}}

\newcommand{\VAR}{\textsc{VAR}}
\newcommand{\RSS}{\textsc{RSS}}
\newcommand{\BIC}{\textsc{BIC}}
\newcommand{\var}{\text{var}}
\newcommand{\pr}{\mbox{pr}}

\def\askip{\vspace{0.2in}}

\begin{document}

\title{\bf High Dimensional and Banded Vector Autoregressions} 
\author{Shaojun Guo$^{\dag}$ \qquad  Yazhen Wang$^\ddag$ \qquad Qiwei Yao$^\star$\\
$^\dag$Institute of Statistics and Big Data, Renmin University of China, \\
Beijing 100872, China\\
$^{\ddag}$Department of Statistics, University of Wisconsin, \\
Madison, WI 53706,USA \\
 $^{\star}$Department of Statistics, London School of Economics, \\
 London, WC2A 2AE, UK\\
sjguo@ruc.edu.cn \ \ yzwang@stat.wisc.edu \ \ q.yao@lse.ac.uk
}

\maketitle

\begin{abstract}
We consider a class of vector autoregressive models with
banded coefficient matrices. The setting represents a type of sparse
structure for high-dimensional time series, though the implied autocovariance
matrices are not banded. The structure is also practically
meaningful when the order of component time series is arranged
appropriately. The convergence rates for the estimated banded autoregressive
coefficient matrices are established.  We also propose
a Bayesian information criterion for determining
the width of the bands in
the coefficient matrices, which is proved to be consistent.
By exploring some approximate banded structure for the auto-covariance
functions of banded vector autoregressive processes, consistent estimators for the
auto-covariance matrices are constructed.
\end{abstract}


\section{Introduction}

The demand for modelling and forecasting high-dimensional time series arises
from 
panel studies of economic,
social and natural phenomena, financial market analysis,
communication engineering and other domains.
When the dimension of time series is even moderately large, statistical
modelling is challenging, as vector autoregressive and moving average models
suffer from lack of identification, over-parameterization and flat likelihood functions.
While pure vector autoregressive models are perfectly identifiable, their
usefulness is often hampered by the lack of proper means of reducing
the number of parameters.

In many practical situations
it is enough to collect the information from neighbour variables, though
the definition of neighbourhoods is case-dependent.
For example, sales, prices, weather indices or electricity consumptions
influenced by temperature depend on those at nearby locations,
in the sense that the information from farther locations may become
redundant given that from neighbours.
See, for example, \citet{CM1997} for a house price example which exhibits such a dependence structure.
In this paper, we propose a class of vector autoregressive models to cater for such dynamic structures.
We assume that the autoregressive coefficient
matrices are banded, i.e., non-zero coefficients form a narrow band along the main diagonal.  The setting specifies explicit autoregression
over neighbour component series only.
Nevertheless, non-zero cross correlations
among all component series may still exist, as the implied
auto-covariance matrices are not banded. This is an effective way to
impose sparse structure, as the number of parameters
in each autoregressive coefficient matrix is reduced from $p^2$ to $O(p)$, where
$p$ denotes the number of time series. In practice, a banded structure may be
employed by arranging the order of component series appropriately.
The ordering can be deduced from subject knowledge aided by statistical
tools such as Bayesian information criterion; see Section \ref{sec52}.
With the imposed banded structure, we propose least squares estimators for
the autoregressive coefficient matrices which attain the convergence rate $(p/n)^{1/2}$
under the Frobenius norm and $(\log p/n)^{1/2}$ under the spectral norm when  $p$ diverges
together with the length $n$ of time series.

In practice the maximum width of the non-zero coefficient bands in the coefficient
matrices, which is called the bandwidth, is unknown. We propose a
marginal Bayesian information criterion to identify the true
bandwidth. It is shown that this criterion leads to consistent
bandwidth determination when both $n$ and $p$ tend to infinity.

We also address the estimation of the autocovariance functions for high-dimensional
banded autoregressive models. Although the autocovariance matrices of a banded process are unlikely
to be banded, they admit some asymptotic banded approximations
when the covariance of innovations is banded.
Because of this property, the band-truncated sample autocovariance matrices
are consistent estimators with the convergence rate $\log(n/\log p)(\log p/n)^{1/2}$,
which is faster than that for the standard  banding covariance estimators \cite[]{bickel2008}.
See also \citet{WP2009}, \citet{bickel2011} and \citet{LL2011} for
the estimation of the banded covariance matrices of time series.

Most existing work on high-dimensional autoregressive models draws inspiration  from recent developments in high-dimensional regression.
For example, \citet{HHC2008} proposed lasso penalization for subset
autoregression.  \citet{HNMK2009} introduced the group sparsity for coefficient matrices and
advocated use of group lasso penalization. A truncated weighted lasso
and group lasso penalization approaches were proposed by \citet{SM2010} and \citet{BSM2015}, respectively, to
explore graphical Granger causality. \citet{BM2015} focused on stable Gaussian processes and investigated the
theoretical properties of $L_1$-regularized estimates of transition matrix in sparse autoregressive models. 
\citet{BVN2011} inferred sparse
causal networks through vector autoregressive processes and proposed a group lasso
procedure.
\citet{KC2012} established oracle inequalities for high-dimensional vector autoregressive models.
\citet{HL2013} proposed an alternative Dantzig-type
penalization and formulated the estimation problem into a linear program.
\citet{CXW2013} studied sparse covariance and precision matrix
in high dimensional time series under a general dependence structure.

\section{Methodology}
\label{sec2}

\subsection{Banded vector autoregressive models}
\label{sec21}
Let $\by_t$ be a $p\times 1$ time series defined by
\begin{eqnarray}
\label{var2}
\by_t=\bA_1 \by_{t-1}+ \cdots+ \bA_d \by_{t-d} + \bep_{t},
\end{eqnarray}
where $\bep_t$ is the innovation at time $t$, $E(\bep_t)=0$ and $\var(\bep_t) = E(\bep_t \bep_t^{\T})= \bSigma_{\bep}$,
and $\bep_t$ is independent of $\by_{t-1}, \by_{t-2}, \ldots$.
Furthermore, all the coefficient matrices $\bA_1, \ldots, \bA_d$ are banded
in the sense that
\begin{equation} \label{a6}
a_{ij}^{(\ell)} =0, \quad  |i-j|> k_0 ,~ \ell = 1,\ldots,d,
\end{equation}
where
$a_{ij}^{(\ell)}$ denotes the $(i,j)$-th element of $\bA_\ell$. Thus the maximum number
of non-zero elements in each row of  $\bA_\ell$ is  the bandwidth $2k_0+1$, and $k_0$ is called the bandwidth parameter. We assume that $k_0\ge 0$ and $d\ge 1$ are fixed integers, and $p \gg k_0,d$. Our goal is to determine $k_0$ and
to estimate the banded coefficient matrices $\bA_1, \ldots, \bA_d$. For simplicity,
we assume that the autoregressive order $d$ is known, as the order-determination problem
has already been thoroughly studied; see, e.g., Chapter 4 of \citet{L2007}.

Under the condition $\det(\bI_p-\bA_1 z - \cdots - \bA_d z^d)\neq 0$ for any $|z|\leq 1$,
model (\ref{var2}) admits a weakly stationary solution $\{ \by_t \}$, where
$\bI_p$ denotes the $p\times p$ identity matrix. Throughout this paper,
$\by_t$ refers to this stationary process. If, in addition, $\bep_t$ is independent and identically distributed,
$\by_t$ is also strictly stationary.

In model (\ref{var2}), we do not require $\var(\bep_t) = \bSigma_{\bep}$
to be banded, but even if it is, the autocovariance matrices are not necessarily banded;
see (\ref{a1}) below. Therefore, the proposed banded model is applicable when
the linear dynamics of each component series depend predominately on its
neighbour series, though there may be non-zero correlations among
all component series of $\by_t$.

\subsection{Estimating banded autoregressive coefficient matrices}
\label{sec22}

Since each row of $\bA_\ell$ has at most $2k_0+1$ non-zero elements,
there are at most $(2k_0+1)d$ regressors in each row
on the right-hand side of (\ref{var2}). For $i=1, \ldots, p$, let $\bbeta_i$ be
the column vector obtained by stacking the non-zero elements in
the $i$-th rows of $\bA_1, \ldots, \bA_d$ together. Let $\tau_i$ denote the
length of $\bbeta_i$. Then
\begin{equation} \label{a4}
\tau_i \equiv \tau_i(k_0) = \Big\{
\begin{array}{ll}
(2k_0+1)d, &  i= k_0 +1,\, k_0 +2,  \ldots,  p-k_0,\\
(2k_0+1-j)d, \quad  & i = k_0 +1-j \;\; {\rm or}\;\; p-k_0+j, \quad j=1, \ldots, k_0.
\end{array}
\end{equation}
Now (\ref{var2}) can be written as
\begin{equation} \label{a2}
y_{i,t} = \bx_{i,t}^{\T} \bbeta_i + \ep_{i,t}, \qquad i=1, \ldots, p,
\end{equation}
where $y_{i,t}$, $\ep_{i,t}$ are respectively the $i$-th component of $\by_t$ and
$\bep_t$ and $\bx_{i,t}$ is the $\tau_i\times 1$ vector consisting of the corresponding
components of $\by_{t-1}, \ldots, \by_{t-d}$.
Consequently, the least squares estimator of $\bbeta_i$ based on (\ref{a2}) is
\begin{equation} \label{a3}
\wh{\bbeta}_i = (\bX_i^{\T}  \bX_i )^{-1} \bX_i^{\T} \by_{(i)},
\end{equation}
where $ \by_{(i)} = (y_{i, d+1}, \ldots, y_{i,n})^{\T}$, and $\bX_i$ is an $(n-d)\times \tau_i$
matrix with $\bx_{i, d+j}^{\T}$ as its $j$-th row.

By estimating $\bbeta_i$, $i=1, \ldots, p$, separately based on (\ref{a3}), we obtain
the least squares estimators $\wh \bA_1, \ldots, \wh\bA_d$ for the coefficient
matrices in (\ref{var2}). Furthermore, the resulting residual sum of squares
is
\begin{equation}
\label{rss-1}
\RSS_i \equiv \RSS_i(k_0) = \by_{(i)}^{\T}\{ \bI_{n-d} - \bX_i (\bX_i^{\T}  \bX_i)^{-1} \bX_i^{\T}\}
\by_{(i)}.
\end{equation}
We write this as a function of $k_0$ to stress that the above estimation presupposes that the bandwidth is $(2k_0 + 1)$ in the sense of (\ref{a6}).

\subsection{Determination of bandwidth}
\label{sec23}

In practice the bandwidth is unknown and we need to estimate $k_0$. We propose to determine $k_0$ based on the
marginal Bayesian information criterion,
\begin{equation}
\label{marginalBIC}
\BIC_i(k) = \log  \RSS_i(k)  + \frac{1 }{ n} d \tau_i(k) C_n \log (p \vee n), \qquad i=1, \ldots, p,
\end{equation}
where $\RSS_i(k)$ and $\tau_i(k)$ are defined, respectively, in (\ref{rss-1}) and (\ref{a4}), $p \vee n = \max(p,n)$,
and $C_n>0$ is some constant which
diverges together with $n$; see Condition 2.
We often take $C_n$ to be $\log \log n$. An estimator for $k_0$ is
\begin{equation} \label{a8}
\wh k = \max_{1\le i \le p}\big\{ \arg\min_{1\le k \le K} \BIC_i(k) \big\},
\end{equation}
where $K\ge 1$ is a prescribed integer. Our numerical study shows that the procedure is insensitive to the choice of $K$ provided $K \ge k_0$. In practice, we often take $K$ to be  $ [n^{1/2}]$ or choose $K$ by checking the curvature of $\BIC_i(k)$ directly.

\noindent
{\bf Remark 1}.
	If the order $d$ is unknown, we can modify the criterion in (\ref{a8}) as follows. Let $\rss_i(k,\ell)$ and
	$\tau_i(k,\ell)$ be defined similarly to (\ref{rss-1}) and (\ref{a4}). The marginal Bayesian information criterion is
	\begin{equation}
	\label{marginalBIC2}
	\tilde{\BIC}_i(k,\ell ) =\log  \RSS_i(k,\ell )  + \frac{1 }{ n}  \tau_i(k,\ell ) C_n \log( p \vee n), \qquad i=1, \ldots, p.
	\end{equation}
	Let $L$ be a prescribed integer upper bound on $d$ and often taken to be $10$ or $[n^{1/2}]$.  Let
	$$
	(\hat{k}_i,\hat{d}_i) = \arg\min_{1\le k \le K, 1 \le \ell \le L} \tilde{\BIC}_i(k,\ell), ~ i =1,\ldots, p,
	$$
	and $\hat{k} = \max_{1\le i \le p} \hat{k}_i$ and $\hat{d} = \max_{1\le i \le p} \hat{d}_i.$
Proposition 1 in the Supplementary Material shows that under Conditions 1--4 in Section 3.1, $\pr(\hat{k} = k_0, \hat{d} = d) \to 1$ as $n$ and $p \to \infty$.

\noindent
{\bf Remark 2}.
	The banded structure of the coefficient matrices $A_1, \ldots, A_d$
	depends on the order of the component series of $\by_t$.
	In principle it is possible to derive a complete data-driven
	method to deduce the optimal ordering which minimizes the bandwidth, but
	such a procedure is computationally
	burdensome for large $p$. For most applications meaningful
	orderings are suggested by practical consideration. We can then calculate
	\begin{equation} \label{2remark}
	\BIC = \sum_{i=1}^p \BIC_i(\wh k)
	\end{equation}
	for each suggested ordering, and choose the ordering which minimizes (\ref{2remark}).
	In expression (\ref{2remark}), $\BIC_i(\cdot)$ and $\wh k$ are defined as in (\ref{marginalBIC}) and (\ref{a8}). Two real data examples in Section 5.2
	indicate that this scheme works well in applications.

\section{Asymptotic properties}
\label{sec3}

\subsection{Regularity conditions}
\label{sec31}
For vector $\bv=(v_1, \ldots, v_j)$ and matrix $\bB=(b_{ij})$, let
$$
\|\bv\|_q = \Big(\sum_{j= 1}^p |v_j|^q\Big)^{1/q} , \quad \|\bv\|_\infty
= \max_{1 \le j \le p} |v_j|, \quad
\|\bB\|_q = \max_{\|\bv\|_q = 1} \|\bB \bv\|_q,  \quad \|\bB\|_F =
\Big(\sum_{i,j} b_{ij}^2\Big)^{1/2},
$$
i.e., $\|\cdot\|_q$ denotes the $\ell_q$ norm of a vector or matrix, and $\|\cdot\|_F$
denotes the Frobenius norm of a matrix.

First we note that  the model (\ref{var2}) can be formulated as,
$$
\widetilde{\by}_t = \widetilde{\bA}\widetilde{\by}_{t-1} + \widetilde{\bep}_t,
$$
where
\begin{equation} \label{a7}
\wt{\by}_t = \left(
\begin{array}{cccc}
\by_{t}\\
\by_{t-1}\\
\vdots\\
\by_{t-d+1}
\end{array}\right), \quad
\wt{\bA} = \left(
\begin{array}{cccc}
\bA_1 &\bA_2 & \cdots & \bA_d\\
\bI_p & \bzero_p & \cdots & \cdots\\
\vdots & \cdots & \vdots & \cdots\\
\bzero & \cdots & \bI_p & \bzero\\
\end{array}\right),  \quad
\wt{\bep}_t = \left(
\begin{array}{cccc}
\bep_{t}\\
\bzero_{p \times 1}\\
\vdots\\
\bzero_{p \times 1}
\end{array}\right).
\end{equation}

Now we list the regularity conditions required for our asymptotic results.

\noindent
{\it Condition} 1.
	For $\wt \bA$ defined in (\ref{a7}), $\|\widetilde{\bA}\|_2 \le C$ and
	$\|\widetilde{\bA}^{j_0}\|_2 \le \delta^{j_0}$, where
	$C>0$, $\delta \in (0, 1)$ and ${j_0}\ge 1$ are constants free of $n$ and $p$,
	and ${j_0}$ is an integer.
	
\noindent
	{\it Condition} 1'. For $\wt \bA$ defined in (\ref{a7}),
	$\|\widetilde{\bA}^{j_0}\|_{2} \le \delta^{j_0}$, $\|\widetilde{\bA}\|_{\infty} \le C$ and
	$\|\widetilde{\bA}^{j_0}\|_{\infty} \le \delta^{j_0}$, where
	$C>0$, $\delta \in (0, 1)$ and ${j_0}\ge 1$ are constants free of $n$ and $p$,
	and ${j_0}$ is an integer.
	
\noindent
{\it Condition} 2.
	Let $a_{ij}^{(\ell)}$ be the $(i,j)$-th element of
	$\bA_\ell$. For each $i=1,\ldots, p$, $|a^{(\ell)}_{i, i+k_0}|$ or $|a^{(\ell)}_{i, i-k_0}|$ is
	greater than $ \{C_n k_0n^{-1}\log (p \vee n) \}^{1/2}$ for
	some $1 \le \ell \le d$, where $C_n \to \infty$ as
	$n \to \infty$. 

\noindent
{\it Condition} 3.
The minimal eigenvalue
	$\lambda_{\min}\{\cov(\by_t)\} \ge \kappa_1$  and $\max_{1 \le i \le
		p}|\sigma_{ii}| \le \kappa_2$ for some positive constants $\kappa_1$ and
	$\kappa_2$ free of $p$, where $\sigma_{ii}$ is the
	$i$-th diagonal element of $\cov(\by_t)$, and
	$\lambda_{\min}(\cdot)$ denotes the minimum eigenvalue.

\noindent
{\it Condition} 4.
	The innovation process $\{\bep_t, \, t=0,\pm 1, \pm 2,\ldots\}$ is
	independent and identically distributed	with zero mean and covariance $\bSigma_{\ep}$. Furthermore, one of the
	two assertions holds:
	\begin{quote}
		(i)~~ $\max_{1 \le i \le p}E(|\varepsilon_{i,t}|^{2q})
		\le C$ and $ p = O(n^\beta)$, where $q>2$, $\beta \in (0, (q - 2)/4)$ and $C>0$
		are some constants free of $n$ and $p$;
		
		(ii) ~~  $\max_{1\le i \le p}E\{\exp(\lambda_0 |\varepsilon_{i,t}|^{2\alpha})\} \le C$
		and $\log p = o\{n^{{\alpha }/(2 - \alpha)}\}$,
		where $\lambda_0>0$, $\alpha \in (0, 1]$  and $C>0$ are constants free of $n$ and $p$.
	\end{quote}

Provided $\{\bep_t\}$ is independent and identically distributed, Condition 1 implies that $\by_t$ is strictly stationary and that for any $j \ge 1$,
$\|\widetilde{\bA}^{j}\|_2 \le C\delta^{j}$ with some constant $C>0$ and
$\delta \in (0,1)$. The independent and identically distributed assumption in Condition 4 is imposed to simplify the proofs but is
not essential. Condition 2 ensures that the bandwidth $(2k_0+1)$ is asymptotically
identifiable, as $\left\{n^{-1}\log (p \vee n)\right \}^{1/2}$ is the minimum order of a non-zero
coefficient to be identifiable; see, e.g., \citet{LC2013}. Condition 3
guarantees that the covariance matrix $\var(y_t)$ is strictly positive definite.
Condition 4 specifies the two asymptotic modes: (i) high-dimensional cases with
$p=O(n^\beta)$, and (ii) ultra high-dimensional cases with $\log p =
o\{n^{{\alpha}/{(2 - \alpha)}}\}$. 

\subsection{Asymptotic theorems}

We first state the consistency of the selector $\wh k$, defined in
(\ref{a8}), for determining the bandwidth parameter $k_0$.

\begin{theorem}
	\label{thm1}
	Under Conditions 1--4, $\pr(\wh{k} = k_0) \to 1$ as
	$n \to \infty$.
\end{theorem}

\noindent
{\bf Remark 3}. In Theorem \ref{thm1}, $k_0$ is assumed to be fixed,
	as in applications small $k_0$ is of particular interest.
	But we can allow the bandwidth parameter $k_0$ to diverge as $n,p \to
	\infty$. To show its consistency, the regularity conditions would need to be
	strengthened. To be specific, if $k_0 \ll C_n^{-1}n/\log (p \vee
	n)$,  $\pr(\wh{k} = k_0) \to 1$ as $n \to \infty$ under
	Conditions 1' and 2--4 in Section 3.1; see the Supplementary Material.

Since $k_0$ is unknown, we replace it by $\wh k$ in the estimation procedure for $\bA_1, \ldots, \bA_d$
described in Section \ref{sec22},  and still denote the resulted estimators by $\wh \bA_1, \ldots, \wh \bA_d$.
Theorem 2 addresses their convergence rates.

\begin{theorem}
	\label{thm2}
	Let Conditions 1--4 hold.
	As $n \to \infty$, it holds for $j=1, \ldots, d$ that
	$$
	\big \|\widehat{\bA}_j - \bA_j\big\|_F = O_P\Big\{(p / n)^{1/2}\Big \},~~
	\big \|\widehat{\bA}_j - \bA_j\big\|_2 = O_P\Big\{ (\log p/ n)^{1/2}\Big\}.
	$$
\end{theorem}
\askip

Conditions 4(i) and 4(ii) impose, respectively, a high moment condition and an
exponential tail condition on the innovation distribution.
Although the convergence rates in Theorem 2 have the same expressions in
terms of $n$ and $p$, due to the different conditions imposed on them in
Conditions 4(i) and 4(ii),
the actual convergence rates are different under the two settings.
For example, Condition  4(i) allows $p$ to grow in the order $n^{\beta}$, which
implies the convergence rate ${(\log n /n)}^{1/2}$ for $\wh{\bA}_j$ under the spectral norm.
On the other hand, Condition  4(ii) may allow $p$ to diverge at the rate
$\exp\{ n^{\alpha/(2-\alpha)- 2\epsilon } \}$ for a small constant $\epsilon>0$, and the implied
convergence rate for $\wh{\bA}_j$ under the spectral norm is
$n^{1/2 + \epsilon -\alpha/(4-2\alpha)}$.

\section{Estimation for auto-covariance functions}
\label{sec4}

For the banded vector autoregressive process $\by_t$ defined by (\ref{var2}), the auto-covariance function
$\bSigma_j = \cov( \by_t, \by_{t+j})$ is unlikely to be banded. For
example for a stationary banded autoregressive process with order 1, it can be shown that
\begin{equation} \label{a1}
\bSigma_0 \equiv \var(\by_t) = \bSigma_{\bep} + \sum_{i = 1}^\infty \bA_1^i
\bSigma_{\bep} (\bA_1^{\T})^i.
\end{equation}
For any banded matrices
$\bB_1$ and $\bB_2$ with
bandwidths $2k_1+1$ and $ 2k_2+1$, respectively, the product  $\bB_1 \bB_2$ is a banded matrix with the
enlarged bandwidth $2(k_1+k_2)+1$ in general.
Thus  $\bSigma_0$ presented in (\ref{a1}) is not a banded matrix.
Nevertheless if $\var(\bep_t) = \bSigma_{\ep}$ is also banded, 
Theorem 3 shows that $\bSigma_j$ can be approximated by
some banded matrices.

\noindent
{\it Condition} 5.
	The matrix $\bSigma_{\bep}$ is banded with bandwidth $2s_0+1$ and
	$\|\bSigma_{\bep}\|_1 \le C < \infty$, where  $C, s_0>0$ are constants independent of $p$,
	and $s_0$ is an integer.

\begin{theorem}
	\label{thm3}
	Let Conditions 1 and 5 hold. For any integers $r, j \ge 0$, there exists a
	banded matrix $ \bSigma_j^{(r)}$ with bandwidth
	$2\{(2r+j)k_0+s_0\} +1$  such that
	\[
	\|\bSigma_j^{(r)} - \bSigma_j\|_2 \le C_1 \, \delta^{2 (r + j)+1},~~
	\|\bSigma_j^{(r)} - \bSigma_j\|_1 \le C_2 \, r\,\delta^{2(r+j) + 1},
	\]
	where $C_1$ and $C_2$ are positive constants independent of $r$ and $p$, and
	$\delta \in (0, 1)$ is specified in Condition 1.
\end{theorem}

Under Condition 5, $\bSigma_0^{(r)} = \bSigma_{\ep} +
\sum_{1 \le i \le r} \bA_1^i \bSigma_{\ep} (\bA_1^{\T})^i$ is a banded matrix
with bandwidth $2(2rk_0+s_0) +1$.  Theorem 3 ensures that the norms of
the difference $\bSigma_0 -\bSigma_0^{(r)} = \sum_{i>r}  \bA_1^i \bSigma_{\ep} (\bA_1^{\T})^i$
admit the required upper bounds.
Theorem 3 also paves the way for estimating $\bSigma_j$ using the banding
method of \citet{bickel2008}, as $\bSigma_j$ can be approximated by a banded matrix
with a bounded error and thus may be effectively treated as a banded matrix.
   To this end, we define the banding operator as follows:
for any matrix $\bH = (h_{ij})$,
$B_r (\bH) = \big\{ h_{ij}I(|i-j|\le r ) \big\}$. Then the banding estimator for $\bSigma_j$
is defined as
\begin{equation}
\wh{\bSigma}_j^{(r_n)} = B_{r_n}( \wh{\bSigma}_j), ~~
\wh{\bSigma}_j = \frac{1}{n} \sum_{t=1}^{n-j} (\by_t - \bar \by ) (\by_{t+j} - \bar \by)^{\T},
~~
\bar{\by} = \frac{1}{n} \sum_{t=1}^n \by_t,
\end{equation}
where $r_n = C \log(n/\log p),
$
and $C>0$ is a constant greater than $\left(-4 \log \delta \right)^{-1}$.
Theorem 4 presents the convergence rates for 
$\wh{\bSigma}_j^{(r_n)}$, which are faster than those in \citet{bickel2008}, due to the approximate banded structure in Theorem 3.

\begin{theorem}
	\label{thm4}
	Assume that Conditions 1--5 hold.
	Then for any integer $j \geq 0$, as $n,p\to \infty$,
	\begin{eqnarray*}
	\|\widehat{\bSigma}_{j}^{(r_n)} - \bSigma_j\|_2 = O_P \Big\{r_n
	{\left({n^{-1}\log p}\right)}^{1/2} + \delta^{2 (r_n+j) + 1}\Big\}
	=O_P \Big\{ \log(n/\log p){\left({n^{-1}\log p}\right)}^{1/2}\Big\},
	\end{eqnarray*}
	and
	$$
         \|\widehat{\bSigma}_{j}^{(r_n)} - \bSigma_j\|_1	
         = O_P \Big\{	\log(n/\log p)\left(n^{-1}\log p \right)^{1/2}\Big\}.
	$$
\end{theorem}

In practice we need to specify $r_n$. An ideal selection would be
$r_{n} = \arg \min_{r} R_j(r)$, where
$$
R_j(r) = E(\|\widehat{\bSigma}_{j}^{(r)} - \bSigma_j\|_1),
$$
but in practice this is unavailable because $\bSigma_j$ is unknown. We replace it
by an estimator obtained via a wild bootstrap. To this end, let
$u_1, \ldots, u_n$ be independent and identically distributed with $E(u_t)= \var(u_t) =1$. A bootstrap estimator
for $\bSigma_j$ is defined as
\[
\bSigma_j^* = {1 \over n} \sum_{t=1}^{n-j} u_t (\by_t - \bar \by ) (\by_{t+j} - \bar \by)^{\T}.
\]
For example, we may draw $u_t$  from the standard exponential distribution.
Consequently the bootstrap estimator for $R_j(r) $ is defined
as
\[
R_j^*(r) = E \big\{ \| B_r(\bSigma_{j}^*) - \wh{\bSigma}_j\|_1 \big|\, \by_1 , \ldots, \by_n \big\}.
\]
We choose $r_{n}$ to minimize $R_j^*(r)$.
In practice we use the approximation
\begin{equation} \label{boots}
R_j^*(r) \approx \frac{1}{q} \sum_{k = 1}^q
\|B_r(\bSigma_{j,k}^*) - \widehat{\bSigma}_{j}\|_1,
\end{equation}
where $\bSigma_{j,1}^*, \ldots, \bSigma_{j,q}^*$ are $q$ bootstrap estimates for $\bSigma_j$, obtained
by repeating the above wild bootstrap scheme $q$ times, and $q$ is a large integer.

\section{Numerical properties}
\label{sec5}

\subsection{Simulations}
\label{sec51}
In this section, we evaluate the finite-sample properties of the proposed methods for the model
$$
\by_t = \bA \by_{t-1} + \bep_t,
$$
where $\{ \bep_t \}$ are independent and $N(0, I_p)$.
We consider two
settings for the banded coefficient
matrix $\bA=(a_{ij})$ as follows:
\begin{description}
	\item {(i)}  $\{ a_{ij}; \; |i-j| \le k_0\}$  are generated independently from $U[-1,1]$.
	Since the spectral norm of $\bA$ must be smaller than 1, we re-scale $\bA$ by
	$ \eta  \bA/\|\bA\|_2 $, where $\eta$ is generated from $U[0.3,1.0)$;
	
	\item {(ii)}  $\{a_{ij}; \; |i-j| < k_0\}$ are generated
	independently from the mixture distribution
	$
	\xi \cdot 0 + (1- \xi) \cdot N(0,1)
	$
	with $\pr(\xi = 1) = 0.4$. The elements $\{a_{ij}; |i-j| = k_0\}$
are drawn independently from $-4$ and  $4$ with probability 0.5 each.
Then $\bA$ is rescaled as in (i) above.
\end{description}
In (ii), there are about $0.4(2k_0-1)p$ zero elements within the band, i.e.,  $\bA$ is sparser than in (i).


We set $n = 200$, $p = 100,
200,400, 800$, and $k_0 = 1, 2, 3, 4$. We repeat each setting
500 times. We only report the results with $K=15$ in (\ref{a8}), as the
results with other values of $K \ge k_0$ are similar.
Table {\ref{table1}} lists the relative frequencies of the occurrence of
the events $\{ \wh k =k\}, \;
\{ \wh k>k_0\}$ and $ \{ \wh k < k_0\}$ over the 500 replications. Overall 
$\wh k$ under-estimates $k_0$, especially when $k_0 = 3$ or 4. In fact when $k_0=4$,
$\wh k$ chose 3 most times. The constraint $\|\bA \| < 1$ makes  most non-zero elements
small or very small when $p$ is large, and that only the coefficients at least as large as
$\surd{\log(p\vee n)/n}$ are identifiable; see Condition 2. Estimation performs better in setting (ii) than in setting (i), as Condition 2
is more likely to hold at the boundaries of the band in setting (ii).

The Bayesian information criterion (\ref{marginalBIC}) is defined for each row separately. One
natural alternative would be
\begin{eqnarray*}
\bic(k) = \sum_{i = 1 }^p \log \rss_{i}(k) + {1 \over n }|\tilde{\tau}(k)| C_n \log (p \vee n),
\end{eqnarray*}
where $\tilde{\tau}(k) = (2p+1)k - k^2 - k $ is the total number of parameters
in the model. This leads to the following estimator for the bandwidth parameter,
\begin{equation} \label{newK}
\tilde{k} = \arg \min_{1 \le k \le K}\bic(k).
\end{equation}
Although
this joint approach can be shown
to be consistent, its finite sample performance, reported in Table {\ref{table2}}, is worse
than that of the marginal Bayesian information criterion (\ref{marginalBIC}), presented in Table {\ref{table1}}.

We also calculate both $L_1$ and $L_2$ errors in estimating the banded
coefficient matrix $\bA$. The means and the standard deviations of the
errors for setting (i) are reported in Table
{\ref{table3}}. Table \ref{table3} also reports
results from estimating  $\bA$ using the true values for the  bandwidth parameter $k_0$.
The accuracy loss 
in estimating $\bA$ caused by unknown $k_0$ is almost negligible.
The results for setting (ii) are similar and are therefore omitted.

To evaluate the estimation performance for the auto-covariance matrices
$\bSigma_0$ and $\bSigma_1$, we set $k_0 = 3$, and the spectral norm of $\bA$
at 0.8.
Furthermore, we let $\bep_t$ be independent and
$N(0, \bSigma_{\bep})$ now, where $\bSigma_{\bep} = \bB\bB^{\T}
$ and $\bB = (b_{ij}), b_{11} = 1, b_{ij} = 0.8 I(|i - j| = 1) + 0.6I(i = j)$, $i > 1$ or $j>1$.
Table \ref{table4} lists the average estimation errors and the  standard
deviations over 100 replications,
measured by matrix $L_1$-norm. We also report Monte Carlo results for
a thresholded estimator and the sample covariance estimator.
For the banded estimator, we choose $r$ to minimize the bootstrap
loss defined in (\ref{boots}) with $q=100$. For the thresholded estimator, the thresholding
parameter is selected in the same manner.
Table {\ref{table4}} shows that
the proposed banding method performs much better than
the thresholded estimator since it adapts directly to the underlying structure,
while the sample covariance performs much worse than both the banding and threshold
methods.

\subsection{Real data examples}
\label{sec52}



Consider first the weekly temperature data across 71 cities in China from 1 January 1990
to 17 December 17 2000, i.e.,  $p = 71$ and $n = 572$.
Fig.\ref{example1-fig2} displays
the weekly temperature of  Ha'erbin, Shanghai
and Hangzhou, showing strong seasonal behavior with period 52 weeks.
Therefore, we set the seasonal period to be 52 and estimate the seasonal
effects by taking averages of the same weeks across different years.
The deseasonalized series, i.e., the original series subtracting estimated
seasonal effects, are denoted by $\{ \, \by_t; \; t = 1,\ldots,572 \, \}$, and each
$\by_t$ has 71 components.

Naturally we would order the 71 cities according to their geographic locations.
However the choice is not unique. For example, we may order the cities from north
to south, from west to east, from northwest to southeast, or from southwest to
northeast. By setting $d=1$, each ordering leads to a different
banded autoregressive model with order 1. We compare
those four models by one-step ahead, and two-step ahead post-sample prediction for the
last 30 data points in the series. To select an optimum model, we compute (\ref{2remark}) for these four orderings. These numerical results and the selected
bandwidth parameters  $\wh k$ are reported in Table \ref{table5}.
Three out of those four models select $\wh k=2$, while the model based on the ordering
from west to east picks $\wh k=4$.
Overall the model based on the ordering from southwest to northeast is preferred, which also has the minimum one-step ahead post-sample predictive errors.
The performances of the four models in terms of the prediction are very close.

Also included in Table \ref{table5}
are the post-sample predictive errors of the sparse autoregressive model with order 1 obtained via lasso by
minimizing
$$
\sum_{t=2}^{n}\|\by_t - \bA \by_{t-1}\|^2  +  \sum_{i,j = 1}^p\lambda_i|a_{ij}|,
$$
where $\lambda_{1},\ldots, \lambda_{p}$ are tuning parameters
estimated by five-fold cross-validation as in  Bickel and Levina (2008).
The prediction accuracy of the sparse model via lasso is comparable to those
of the banded autoregressive models, though slightly worse, especially for the two-step ahead prediction.
However the lack of any structure in the estimated sparse coefficient
matrix $\wt{\bA}$, displayed in Fig.\ref{example1-fig3}(b),
makes such fits difficult to interpret. In contrast, the banded coefficient
matrix,  depicted in Fig.\ref{example1-fig3}(a), is attractive. 

\askip

As a second example, we consider the daily sales of a clothing
brand in 21 provinces in China from 1 January 2008 to 9 December 2012, i.e., $n = 1812, \; p = 21$.
Fig.\ref{example2-fig1} plots the relative geographical positions of 21 provinces and province-level municipalities.
We first subtract each of the 21 series by its mean.
Similar to the example  above, we order the 21 provinces according to the four different
geographic orientations, and fit a banded autoregressive model with order 1 for each ordering. The selected bandwidth parameters, the values according to (\ref{2remark}) and the post
sample prediction errors for the last 30 data points in the series are reported in
Table \ref{table6}. We also rank the series according to their geographic distances
to Heilongjiang, the most northwestern province; see Fig.\ref{example2-fig1}.
This results in a different ordering to that from north to south.
Table \ref{table6} indicates that the minimum bandwidth parameter $\wh k$
is 3, attained by
the ordering based on the distances to Heilongjiang, followed by $\wh
k=4$ attained by the north-to-south ordering. The post-sample prediction performances of
those two models are almost the same, and are better than those of the other
three banded models and the sparse autoregressive model. 

The ordering based on the direction from
northwest to southeast leads to $\wh k =12$. Therefore the corresponding banded model has 21 regressors for some components according to (3), i.e., no banded structure is observed
in this case.  Fig.\ref{example2-fig1}  indicates that the ordering from
northwest to southeast puts together some provinces which are distance
away from each other. Hence this is certainly a wrong ordering  as far as
the banded autoregressive structure is concerned.

The estimated coefficient matrix $\wh A$ for the banded vector autoregressive model with order 1 based on the distances
to Heilongjiang and the estimated $\wt A$ by lasso for the autoregressive model with order 1 are plotted in
Fig.\ref{example2-fig3}. The banded model facilitates an easy interpretation, i.e.,
the sales in the neighbour provinces are closely associated with each other. The lasso
fitting cannot reveal this phenomenon.

\section*{Acknowledgements}
We are grateful to the Editor, the Associate Editor and two referees for their
insightful comments and valuable suggestions, which lead to significant improvement of  our article.
This research was supported in part by Natural National Science of Foundation of China, 
National Science of Foundation of the United States of America 
and Engineering and Physical Sciences Research Council of the United Kingdom.
 This paper was completed when Shaojun Guo was Research
Fellow at London School of Economics and Assistant Professor at Chinese
Academy of Sciences.

\section*{Supplementary Material}
Supplementary material available at {\it Biometrika} online includes proofs of Theorems 1-4, the consistency of generalized Bayesian information criterion defined by
(9) in Section 2.3 and the consistency of the marginal Bayesian information criterion in the
setting $k_0 \to \infty$, as well as the detailed proofs of all the lemmas in
this paper.


\begin{table}
	\def~{\hphantom{0}}
	\caption{Relative frequencies (\%) for the occurrence of the events $\{ \wh k =k\}, \;
		\{ \wh k>k_0\}$ and $ \{ \wh k < k_0\}$ in a simulation study with $500$ replications, where $\wh k$ is defined
		in (\ref{a8}).}
	 {
		\begin{tabular}{c|c|ccc|ccc}
			\hline \hline
			&  & & Setting (i) &   & & Setting (ii) &\\
			\hline
			&  &  $\{\hat{k} = k_0\}$ & $\{\hat{k} > k_0\}$& $\{\hat{k} < k_0\}$
			&  $\{\hat{k} = k_0\}$ & $\{\hat{k} > k_0\}$& $\{\hat{k} < k_0\}$\\
			\hline
			& $k_0 =1$ & 82  & 17 & 1
			& 98   & 2 &    0 \\
			$p = 100$ & $k_0 =2$ & 87  & 8 & 5
			& 95  & 3 & 2    \\
			& $k_0 =3$ & 73  & 6  & 21
			& 83 & 2 & 15\\
			& $k_0=4$  & 55  & 14  & 31
			& 64  & 2 & 34\\
			\hline
			& $k_0 =1$ & 91  &  9 & 0
			& 97  &  3 & 0\\
			$p = 200$ & $k_0 =2$ & 89  &  4 & 7
			& 93  &  2 & 5    \\
			& $k_0 =3$ & 65  &  3 & 32
			& 83  &     0  & 17\\
			& $k_0 =4$ & 54  &  1 & 45
			& 63  &  2 & 35 \\
			\hline
			& $k_0 =1$ & 95  &  5 & 0
			& 99  &  1 & 0\\
			$p = 400$ & $k_0 =2$ & 87  &  2 & 11
			& 90  &  1 & 9\\
			& $k_0 =3$ & 66  &  2 & 32
			& 76  &  1 & 23 \\
			& $k_0 =4$ & 45  &  1 & 54
			& 60  &  0     & 40\\
			\hline
			& $k_0 =1$ & 97  & 3 & 0
			& 100  &   0   & 0\\
			$p = 800$ & $k_0 =2$ & 86  & 1 & 13
			& 91  & 1 & 8  \\
			& $k_0 =3$ & 59  & 1 & 40
			& 67  & 1 & 32\\
			& $k_0 =4$ & 40  & 0 & 60
			& 52  & 0 & 48\\
			\hline \hline
		\end{tabular}
}
\label{table1}
\end{table}

\begin{table}
	\caption{Relative frequencies (\%) for the occurrence of the events $\{ \wt k =k\}, \;
		\{ \wt k>k_0\}$ and $ \{ \wt k < k_0\}$ in a simulation study with
		$500$ replications, where $\wt k$ is defined in
		(\ref{newK}).}{
		\begin{tabular}{c|c|ccc|ccc}
			\hline \hline
			&    &                 &    Setting (i)   &
			&                 &    Setting (ii)   & \\
			\hline
			&  &  $\{\wt{k} = k_0\}$ & $\{\wt{k} > k_0\}$& $\{\wt{k} < k_0\}$
			&  $\{\wt{k} = k_0\}$ & $\{\wt{k} > k_0\}$& $\{\wt{k} < k_0\}$\\
			\hline \hline
			& $k_0 =1$ & 64  &  0  &  36
			& 88  &  0  &  12\\
			$p = 100$ & $k_0 =2$ & 42  &  0  &  58
			& 63 &  0  &  37\\
			\hline
			& $k_0 =1$ & 56 &  0  &  44
			& 84 &  0  &  16\\
			$p = 200$ & $k_0 =2$ & 32 &  0  &  68
			& 55 &  0  &  45\\
			\hline
			& $k_0 =1$ & 48 &  0  &  52
			& 83 &  0  &  17\\
			$p = 400$ & $k_0 =2$ & 23 &  0  &  77
			& 45 &  0  &  55\\
			\hline
			& $k_0 =1$ & 44 &  0  &  56
			& 76 &  0  &  24\\
			$p = 800$ & $k_0 =2$ & 11 &  0  &  89
			& 41 &  0  &  59\\
			\hline \hline
		\end{tabular}}
	\label{table2}			
\end{table}

\begin{table}
	\caption{ Means ($\times 10^2$) with their corresponding standard deviations ($\times 10^2$) in parentheses of the errors
		in estimating $\bA$ under setting ($i$) in a simulation study with $n=200$ and $500$ replications.}
	\label{table3}	
     \begin{center}
		\begin{tabular}{c|c|cc|cc}
			\hline \hline
			&    & \multicolumn{2}{c|}{ With estimated $k_0$ }   &  \multicolumn{2}{c}{With true $k_0$}  \\
			\hline
			$p$ &    &  $\|\hat{\bA} - \bA\|_1$ &  $\|\hat{\bA} - \bA\|_2$
			&  $\|\hat{\bA} - \bA\|_1$ &  $\|\hat{\bA} - \bA\|_2$ \\
			\hline
			&$k_0 = 1$ &  38 (6) &   27 (3) &  37 (5) &  27 (3) \\
			$p=100$  &$k_0 = 2$ &  54 (6) &   33  (3) &  53 (5) &  33 (3) \\
			&$k_0 = 3$ &  70 (8) &   39 (4) &  69 (7) &  38 (3) \\
			&$k_0 = 4$ &  85 (10) &   43 (5) &  85 (8) &  43 (3) \\
			\hline
			&$k_0 = 1$ &  40 (6) &  28 (3) &   40 (5) &  28 (3) \\
			$p=200$   &$k_0 = 2$ &  58 (7) &  35 (3) &   58 (6) &  35 (3) \\
			&$k_0 = 3$ &  74 (8) &  40 (4) &   74 (6) &  40 (3) \\
			&$k_0 = 4$ &  90 (11) &  46 (5) &   88 (7) &  45 (3) \\
			\hline
			&$k_0 = 1$ &  43 (5) &  30 (3) &   42 (4) &  30 (3) \\
			$p=400$   &$k_0 = 2$ &  60 (6) &  36 (3) &   60 (5) &  36 (3) \\
			&$k_0 = 3$ &  77 (8) &  42 (4) &   76 (6) &  42 (3) \\
			&$k_0 = 4$ &  95 (14) &  48 (7) &   93 (7) &  46 (3) \\
			\hline
			&$k_0 = 1$ & 44 (4) & 31 (2) & 44 (4)  &  31 (2) \\
			$p=800$  &$k_0 = 2$ & 63 (5) & 37 (3) & 62 (5)  &  37 (2)  \\
			&$k_0 = 3$ & 81 (9) & 43 (5) & 80  (6)  &  43 (2) \\
			&$k_0 = 4$ & 98 (14) & 49 (7) & 96 (7)  &  47 (2)  \\
			\hline \hline
		\end{tabular}
	\end{center}		
\end{table}


\begin{table}
	\caption{Means with their corresponding standard deviations in parentheses of the errors in estimating
		autocovariance matrices in a simulation study with $n=200$ and $100$ replications. }{
		\begin{tabular}{c|ccc|ccc}
			\hline \hline
			&  & $\|\widehat{\bSigma}_{n,0} - \bSigma_0\|_1$  &  &
			& $\|\widehat{\bSigma}_{n,1} - \bSigma_1\|_1$  &    \\
			\hline \hline
			&  Banding        & Thresholding &  Sample &  Banding        & Thresholding &  Sample \\
			\hline
			&                & Matrix $L_1$-Norm   &    & & Matrix $L_1$-Norm   &\\
			$p = 100$ &   2.1 (0.04) &   2.6 (0.02) &  14 (0.07)
			&   2.9 (0.03) &   3.5 (0.04) &  14 (0.07) \\
			$p = 200$ &   2.7  (0.04) &   3.4 (0.03) &  29 (0.02)
			&   3.1 (0.03) &   4.2 (0.04) &  30 (0.02) \\
			$p = 400$ &   2.3 (0.02) &   2.9 (0.02) &  55 (0.02)
			&   2.8 (0.03) &   3.7 (0.02) &  55 (0.02) \\
			$p = 800$ &   2.7 (0.03) &   3.4 (0.02) &  112 (0.03)
			&   2.9 (0.03) &   3.9 (0.03) &  110 (0.04)  \\
			\hline \hline
			&               & Spectral Norm  &
			&               & Spectral Norm  & \\
			$p = 100$ & 1.1 (0.01)   & 1.4 (0.02) &   4.0 (0.07)
			& 1.4 (0.01)   & 1.7 (0.02) &   3.7 (0.02)\\
			$p = 200$ & 1.3 (0.03)   & 1.7 (0.02) &   6.5 (0.03)
			& 1.5 (0.01)   & 1.9 (0.01) &   6.1 (0.02)\\
			$p = 400$ & 1.2 (0.01)   & 1.6 (0.01) &   10 (0.03)
			& 1.3 (0.01)   & 1.9 (0.01) &   9.2 (0.02)\\
			$p = 800$ & 1.4 (0.02)   & 1.8 (0.01) &   17 (0.03)
			& 1.4 (0.01)   & 2.3 (0.02) &   15 (0.03) \\
			\hline \hline
		\end{tabular}
	}
\label{table4}		
\end{table}

\begin{table}
	\caption{Results of Example 1: Estimated bandwidth parameters, Bayeysian information criterion values and average one-step-ahead and two-step-ahead post-sample predictive errors over 71 cities with their corresponding standard errors in parentheses.}
	\label{table5}		
	\begin{center}
		\begin{tabular}{l|c|c|c|c}
			\hline \hline
			Ordering & $\hat{k}$ &  BIC & One-step ahead & Two-step ahead	\\
			\hline \hline
			north to south  &	2 & 552.5 & 1.543 (1.170) & 1.622 (1.245)	\\
			west to east  &	4 & 555.9 & 1.545 (1.152) & 1.602 (1.247) \\		
			northwest to southeast  &	2 & 552.4 & 1.552 (1.167) & 1.624 (1.249) \\		
			southwest to northeast  &	2 & 551.9 & 1.538 (1.160) & 1.617 (1.253) \\		
			\hline \hline
			Lasso &	- & - & 1.545 (1.172) & 1.632 (1.250) 	\\			
			\hline \hline					
		\end{tabular}
	\end{center}
\end{table}

\begin{table}
	\caption{Results of Example 2: Estimated bandwidth parameters, Bayeysian information criterion values and average one-step-ahead	and two-step-ahead post-sample predictive errors over 21 provinces with their corresponding standard errors in parentheses.}
	 \label{table6}			
	\begin{center}
		\begin{tabular}{l|c|c|c|c}
			\hline \hline
			Ordering & $\hat{k}$ &  BIC & One-step ahead & Two-step ahead	\\
			\hline \hline
			north to south &  4 & 114.9 & 0.314 (0.377) & 0.407 (0.386)\\
			west to east &  7 & 115.2 &  0.323 (0.363) & 0.409 (0.386)\\
			northwest to southeast & 12 & 115.2 &  0.322 (0.361) & 0.409 (0.395)\\
			southwest to northeast & 5  & 115.1 &  0.316 (0.374) & 0.407 (0.385) \\
			distance to Heilongjiang & 3  & 114.7 &  0.313 (0.378) & 0.407 (0.386) \\
			\hline \hline
			Lasso &	- & - &0.322 (0.362) & 0.410 (0.393) 	\\			
			\hline \hline					
		\end{tabular}
	\end{center}
\end{table}

\begin{figure}
	\begin{center}
		\includegraphics[height= 80mm,width= 140mm]{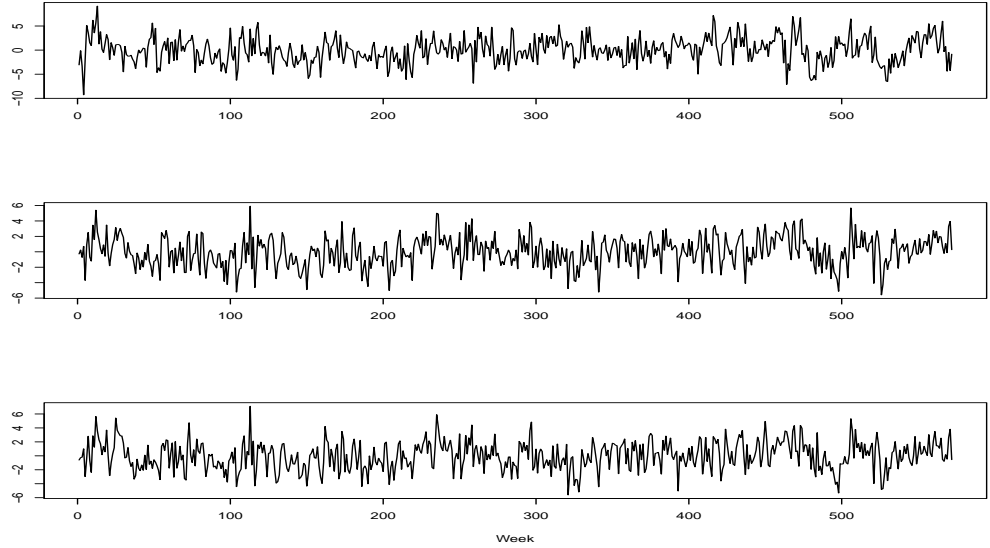}
		\caption{Deseasonalized weekly temperature in degrees Celsius $(^{\circ} \mbox{C})$ from January  1990 to December 2000, where Ha'erbin, Shanghai and Nanjing correspond to the plots from top to bottom.}
		\label{example1-fig2}
	\end{center}
\end{figure}

\begin{figure}	
	\begin{center}
		\includegraphics[height= 75mm,width= 160mm]{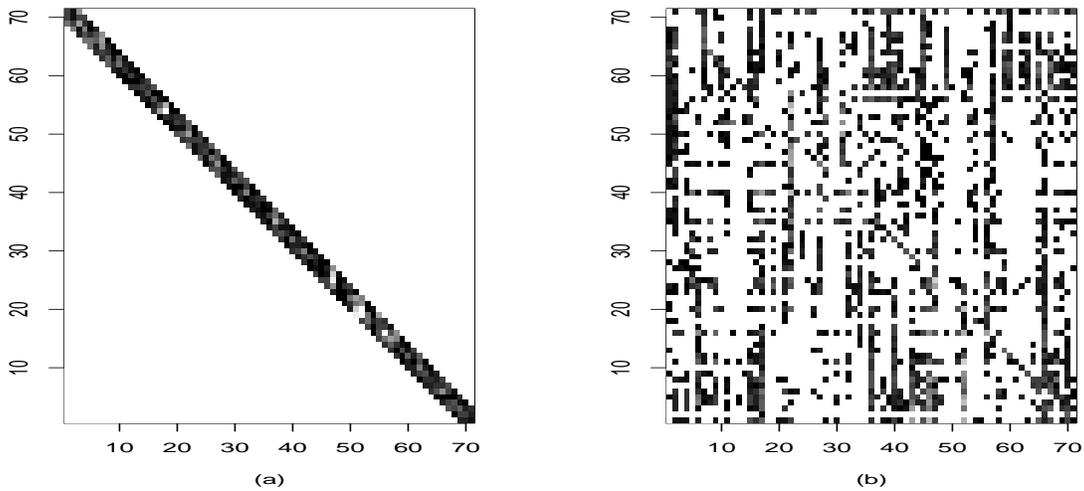}
		\caption{Example 1: (a) Estimated banded coefficient matrix $\wh \bA$
			for the model based on the ordering from southwest to northeast, and (b)
			estimated sparse coefficient matrix $\wt \bA$
			by lasso. White points represent zeros entries and gray or black points represent nonzero entries. The larger the absolute value of a coeffcient is, the darker the colour is.}
		\label{example1-fig3}				
	\end{center}
\end{figure}

\begin{figure}
	\begin{center}	
		\includegraphics[height= 90mm,width= 80mm]{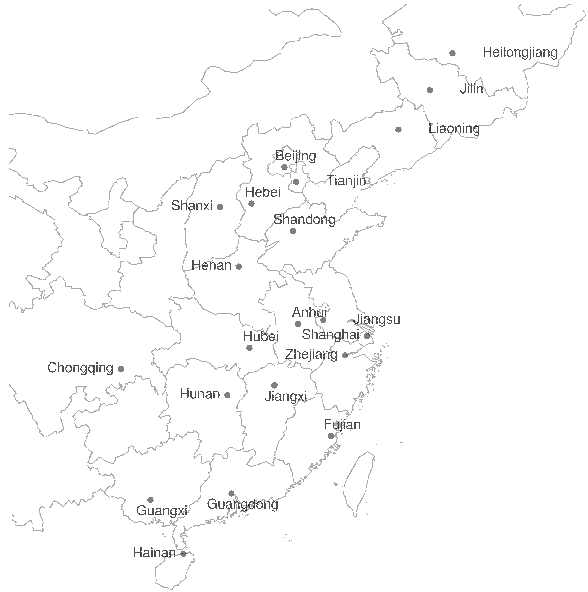}	
		\caption{Location plot of 21 provinces and province-level municipalities 		
			in China, where Shanghai is a province-level municipality, and Ha'erbin, Hangzhou and Nanjing are the capitals of Heilongjiang, Zhejiang, and Jiangsu provinces, respectively. 		
		}
		\label{example2-fig1}
	\end{center}
\end{figure}

\begin{figure}
	\begin{center}
		\includegraphics[height= 70mm,width= 150mm]{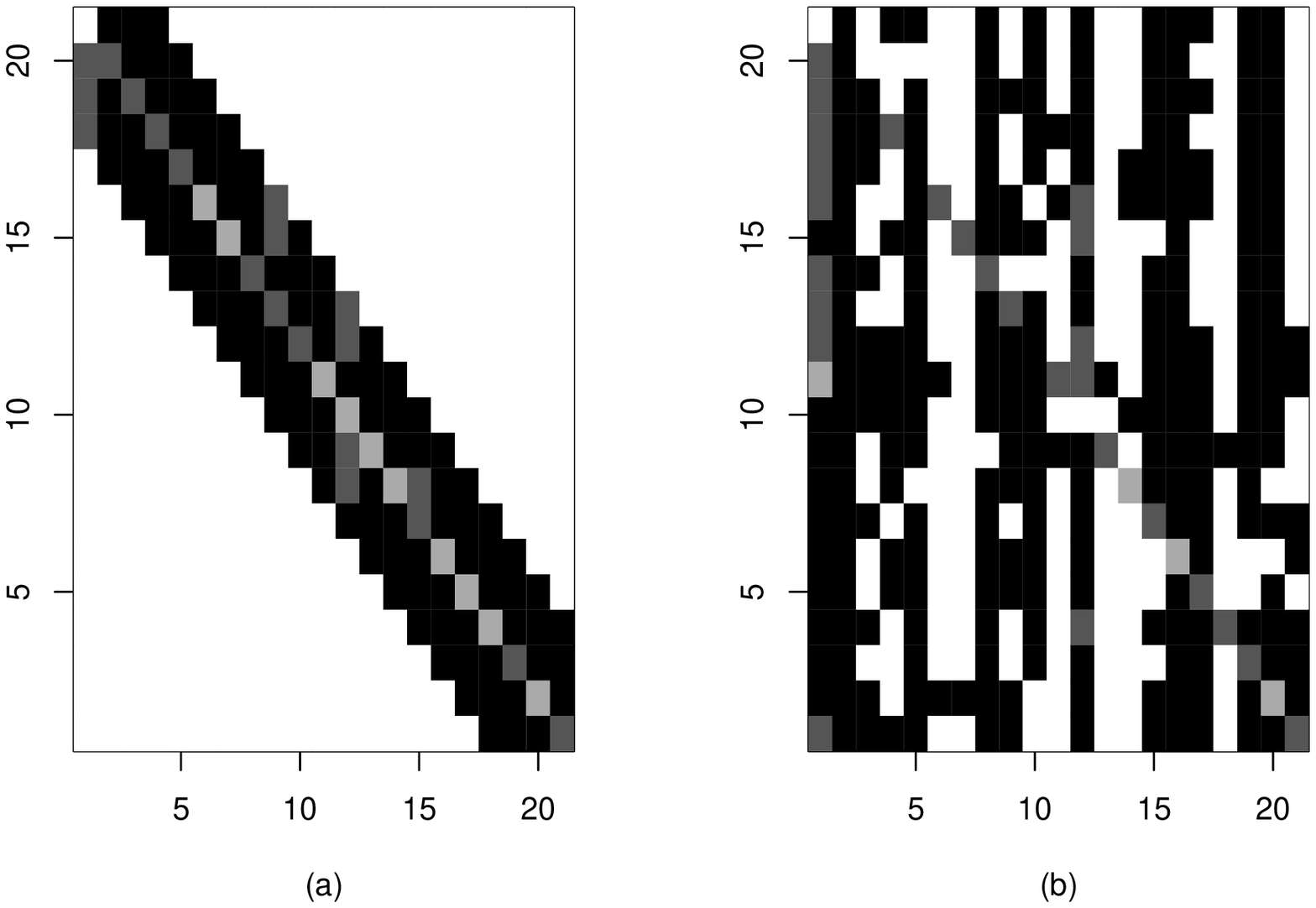}
		\caption{Example 2: (a) Estimated banded coefficient matrix $\wh \bA$
			for the model based on the ordering using distances to
			Heilongjiang, and (b) estimated sparse coefficient matrix $\wt \bA$
			by lasso. White points represent zeros entries and gray or black points represent nonzero entries. The larger the absolute value of a coeffcient is,
			the darker the colour is.}
		\label{example2-fig3}				
	\end{center}
\end{figure}

\newpage
\begin{center}
	{\Large \bf Supplementary Material: High Dimensional and Banded Vector Autoregressions}
\end{center}

\begin{abstract}
This supplementary material is organized as follows. We provide the detailed proofs of Theorems 1-4, respectively, in Sections A.1-A.4. Section A.5
presents Proposition 1 and its proof, showing the consistency of generalized Bayesian Information criterion stated in Remark 1 in the paper. In Section A.6, we present
the consistency of the marginal Bayesian information criterion selector $\hat{k}$ in a more general setting when $k_0 \to \infty$. 
Some technical lemmas and their proofs are collected in Section A.7. Section A.8 presents some additional simulation results.
\end{abstract}

\setcounter{section}{0}
\renewcommand{\thesection}{A.\arabic{section}}

\setcounter{equation}{0}
\renewcommand{\theequation}{A.\arabic{equation}}

\section{Proof of Theorem 1}
Without loss of generality, we consider the \VAR(1) model with $\|\bA\|_1 \le \delta < 1$.
	Our goal is to prove that
	$ \pr(\hat{k}= k_0) \to 1$, i.e., $ \pr(\hat{k} \neq k_0) \to 0$. If $\hat{k} \neq k_0$, then either $\hat{k} > k_0$
	or $\hat{k} < k_0$ holds. Hence it suffices to show that
	$\pr(\hat{k} < k_0) \to 0$ and $\pr(\hat{k} > k_0) \to 0$.
	Our proof follows the arguments in \citet{WLL2009}. 
	
	Consider the first case. Observe that
	$
	\pr(\hat{k} < k_0) \le \pr(\hat{k}_{i} < k_0)
	$
	for some $ i \in \{1,\ldots,p\}$ and the event $(\hat{k}_i < k_0)$ imply
	$
	\{\min_{k < k_0}\bic_{i}(k) < \bic_{i}(k_0)\}.
	$
	To prove $\pr(\hat{k} < k_0) \to 0$, we only need to show that
	$$
	\pr\{\min_{k < k_0}\bic_{i}(k) < \bic_{i}(k_0)\} \to 0
	$$
	for some $i$. Suppose that we have shown that there exists a constant $\eta >0$ and an event $\mathcal{A}_n$ such that
	$
	\pr(\mathcal{A}_n) \to 1
	$
	as $n \to \infty$ and on the event $\mathcal{A}_n$,
	\be
	\label{append-1}
	\rss_i(k) - \rss_i(k_0) \ge \eta  \rss_{i}(k_0)  (a_{i,i-k_0}^{2} + a_{i,i+k_0}^{2}),
	\ee
	for sufficiently large $n$, where $a_{j,k}$ is the $(j,k)$-element of $\bA_1$. 
	On the event $\mathcal{A}_n$ with large $n$,
	$
	\log \rss_i(k) - \log \rss_i(k_0)
	\ge \log\{1 + \eta (a_{i,i-k_0}^{2} + a_{i,i+k_0}^{2})\}.
	$
	Note that $\log (1 + x) \ge \min(0.5x,\log 2)$ for any $x > 0$.
	Consequently, with probability tending to one,
	$\log \rss_i(k) - \log \rss_i(k_0)$ can be further bounded below by
	$
	\min\{0.5 \eta (a_{i,i-k_0}^{2} + a_{i,i+k_0}^{2}),\log 2\}.
	$
	Condition 3 implies that for some $i^* \in \{1,\ldots,p\}$, $a_{i^*,i^*-k_0}^2 + a_{i^*,i^*+k_0}^2 \gg C_n\log p/n$ as $n \to \infty$.  Hence, it follows that, with probability tending to 1,
	\begin{eqnarray*}
		\min_{k < k_0}\bic_{i^*}(k) -  \bic_{i^*}(k_0) &>& \min\{0.5  \eta(a_{i^*,i^*-k_0}^{2} + a_{i^*,i^*+k_0}^{2}),\log 2\} \\
		&& - C_n k_0 n^{-1} \log (p \vee n) > 0,
	\end{eqnarray*}
	where $p\vee n = \max(p,n)$. Hence,
	$
	\pr\{\min_{k < k_0}\bic_{i^*}(k) < \bic_{i^*}(k_0)\} \to 0
	$
	and thus $\pr(\hat{k} < k_0) \to 0$.
	
	Let us prove (\ref{append-1}). For $k < k_0$, denote $\bH_{i,k} = \bX_{i,k}\big(\bX_{i,k}^{\T}\bX_{i,k}\big)^{-1}\bX_{i,k}^{\T}$, $\bX_{i,k_0} = (\bS_{i,k}^{(1)},\bX_{i,k}, \bS_{i,k}^{(2)})$ and $\bbeta_{i,k_0} = (\bb_{i,1}^{\T},\bbeta_{i,k}^{\T}, \bb_{i,2}^{\T})^{\T}$, where $\bX_{i,k}$ is defined similar to (4) in Section 2.2 except that $k_0$ is replaced by $k$. Then
	$
	\rss_i(k)= \by_{(i)}^{\T}\big(\bI_{n-1} - \bH_{i,k}\big)\by_{(i)},
	$
	and by Lemma 5 (ii) or Lemma 6 (ii), we have
	$$
	\rss_i(k) - \rss_i(k_0) = (\bb_{i,1}^{\T},\bb_{i,2}^{\T})(\bS_{i,k}^{(1)},\bS_{i,k}^{(2)})^{\T}(\bI_{n-1} - \bH_{i,k})(\bS_{i,k}^{(1)},\bS_{i,k}^{(2)}) \left(
	\begin{array}{cc}
	\bb_{i,1} \\
	\bb_{i,2}
	\end{array}
	\right)
	+ o_P(1).
	$$
	From Lemma 5 (i) or Lemma 6 (i) and Lemma 7, there exists a small constant $ \eta >0$ such that, with probability tending to one,
	$$
	\lambda_{\min}\big\{(\bS_{i,k}^{(1)},\bS_{i,k}^{(2)})^{\T}(\bI - \bH_{i,k}) (\bS_{i,k}^{(1)},\bS_{i,k}^{(2)})\big\} > \eta (1 +  \eta)
	n\sigma_i^2,
	$$
	and
	$\rss_i(k_0) \le n\sigma_i^2 (1 + \eta)$. Therefore, (\ref{append-1}) follows.
	
	Now we turn to the overfitting case, i.e., $\pr(\hat{k} > k_0) \to 0$. For $k > k_0$, set
	\begin{eqnarray*}
		\bX_{i,k} = (\bS_{i,k}^{(1)},\bX_{i,k_0}, \bS_{i,k}^{(2)}), \bbeta_{i,k} = (\bzero^{\T},\bbeta_{i,k_0}^{\T}, \bzero^{\T})^{\T},
		\bS_{i,k} = (\bS_{i,k}^{(1)},\bS_{i,k}^{(2)}),
	\end{eqnarray*}
	and $\widetilde{\bS}_{i,k} = \big(\bI_{n-1} -\bH_{i,k_0}\big)\bS_{i,k}$. Let $\eta $ be an arbitrary but fixed positive constant and define
	$$
	\mathcal{B}_n = \Big\{\inf_{k_0 \le k \le K}{\inf_{1 \le i \le p}} \frac{\rss_i(k)}{ n\sigma_i^2}
	> (1 - \eta)\Big\},
	$$
	$$
	\mathcal{C}_n = \underset{k_0 \le k\le K}{\bigcup_{1 \le i \le p}}\Big\{\lambda_{\min}^{-1}(n^{-1}
	\widetilde{\bS}_{i,k}^{\T}\widetilde{\bS}_{i,k}) < \kappa_1^{-1}(1 + \eta),\sup_{1 \le j \le k - k_0}\left| \left(n^{-1}\bS_{i,k}^{T} \bS_{i,k} \right)_{jj}\right| <  \kappa_2(1 + \eta)\Big\}.
	$$
	We first give an upper bound on $\rss_i(k_0) - \rss_i(k)$ for $k > k_0$. For each $i$, $\rss_i(k)$ can be rewritten as
	$$
	\rss_i(k) = \inf_{\bb}\|\by_{(i)} - \bX_{i,k} \bb\|^2 =
	\inf_{\bb_{1}, \bb_{2}}\|\by_{(i)} - \bX_{i,k_0} \bb_1 - \bS_{i} \bb_2 \|^2.
	$$
	It can be verified that $\rss_i(k_0) = \|( \bI_{n-1} - \bH_{i,k_0})\by_{(i)}\|^2$ and
	$$
	\rss_i(k) = \rss_i(k_0) - \|\widetilde{\bS}_{i}^{(k)}\wh{\bb}_2\|^2,
	$$
	where $\wh{\bb}_{2} = \left(\widetilde{\bS}_{i,k}^{T}\widetilde{\bS}_{i,k}\right)^{-1}\widetilde{\bS}_{i,k}^{T}\bbe_{(i)}$. Then on the event $\mathcal{C}_n$ we have
	\begin{eqnarray*}
		\rss_i(k_0) - \rss_i(k) &=&		
		\bbe_{(i)}^{\T}\widetilde{\bS}_{i,k}
		(\widetilde{\bS}_{i,k}^{\T}\widetilde{\bS}_{i,k})^{-1}
		\widetilde{\bS}_{i,k}^{\T}\bbe_{(i)} \nonumber\\
		&\le & \kappa_1^{-1} (1 + \eta) |\tau_i(k) - \tau_i(k_0)| \sup_{j,k \le p}\big|n^{-1/2}\bbe_{(j)}^{\T}(\bI_{n-1}- \bH_{i,k_0})\bx_{(k)}\big|^2.
	\end{eqnarray*}
	Define
	$$
	\mathcal{D}_n = \Big\{ \sup_{j,k \le p}\big|n^{-1/2}\bbe_{(j)}^{\T}(\bI_{n-1}- \bH_{i,k_0})\bx_{(k)}\big|^2 \sigma_i^{-2}< \frac{\kappa_1 (1- \eta)} {1 + \eta}C_n \log (p \vee n) \Big\}.
	$$
	On the set $\mathcal{B}_n \cap \mathcal{C}_n \cap \mathcal{D}_n$, for all $k$ with $k_0 \le k \le K$,
	\bs
	\rss_i(k_0) - \rss_i(k)
	&<&  \sigma_i^2 (1 - \eta ) |\tau_i(k) - \tau_i(k_0)| C_n \log (p \vee n)  \\
	&<& \rss_i(k) C_n |\tau_i(k) - \tau_i(k_0)| n^{-1}\log (p \vee n) .
	\es
	Note that $\log (1 + x) \le x$ for any $x > 0$. Hence, for all $k$ with $k_0 < k \le K$, on the set $\mathcal{B}_n \cap \mathcal{C}_n \cap \mathcal{D}_n$,
	\bs
	\bic_i(k) - \bic_i(k_0) &=& \log \rss_i(k) - \log \rss_i(k_0) + C_n|\tau_i(k) - \tau_i(k_0)| n^{-1}\log (p \vee n) \\
	& \ge & - \left\{\rss_i(k_0) -\rss_i(k)\right\} \left\{\rss_i(k)\right\}^{-1} \\
	&& + C_n|\tau_i(k) - \tau_i(k_0)| n^{-1}\log (p \vee n)  > 0,
	\es
	which indicates that over the set $\mathcal{B}_n \cap \mathcal{C}_n \cap \mathcal{D}_n$, we have that $\hat{k} \le k_0$. To prove that $\pr(\hat{k} > k_0) \to 0$, it suffices to show that $ \pr\left\{\big(\mathcal{B}_n \cap \mathcal{C}_n \cap \mathcal{D}_n\big)^c  \right\} \to 0$. In fact,
	it follows from Lemma 7 and Lemma 5 or 6 (i) that $\pr\left(\mathcal{B}_n^c\right) \to 0$ and $\pr\left(\mathcal{C}_n^c\right) \to 0$. It remains to show that $\pr\left(\mathcal{D}_n^c\right) \to 0$. Let $\bSigma_{i,k} = n^{-1}E\left(\bX_{i,k}^{\T}\bX_{i,k}\right)$, $\wh{\bSigma}_{i,k} = n^{-1}\bX_{i,k}^{\T}\bX_{i,k}$, where $E(X)$ denotes the expectation of $X$. Set $\widetilde{\bH}_{i,k} = n^{-1}\bX_{i,k} \bSigma_{i,k}^{-1}\bX_{i,k}^{\T}$, and $\widetilde{\bx}_{(k)} =(\bI_{n-1} - \widetilde{\bH}_{i,k}) \bx_{(k)}$.
	On the event $\mathcal{C}_n$, we obtain that
	\begin{eqnarray*}
		&&\sup_{j,k \le p}\big|\bbe_{(j)}^{\T}(\bI_{n-1}- \bH_{i,k_0})\bx_{(k)}\big| \\
		&& \le \sup_{j,k \le p} \big|\bbe_{(j)}^{\T}\widetilde{\bx}_{(k)}\big | +
		\sup_{j,k \le p}\big |\bbe_{(j)}^{\T} (\bH_{i,k_0} - \widetilde{\bH}_{i,k_0})\bx_{(k)}\big| \\
		&& \le \sup_{j,k \le p} \big|\bbe_{(j)}^{\T}\widetilde{\bx}_{(k)}\big | +
		n^{-1}\sup_{j,k \le p}\|\bbe_{(j)}^{\T} \bX_{i,k_0}\|_2 \|\bSigma_{i,k_0}^{-1}\|_2\|\wh{\bSigma}_{i,k_0}^{-1}\|_2 \|\wh{\bSigma}_{i,k_0} - \bSigma_{i,k_0}\|_2 \|\bX_{i,k_0}^{\T}\bx_{(k)}\|_2 \\
		&& \le \sup_{j,k \le p} \big|\bbe_{(j)}^{\T}\widetilde{\bx}_{(k)}\big | +
		k_0 \kappa_1^{-2} k_2(1+ \eta)^2\sup_{j,k \le p} \big|\bbe_{(j)}^{\T}\bx_{(k)}\big | \cdot \|\wh{\bSigma}_{i,k_0} - \bSigma_{i,k_0}\|_2,
	\end{eqnarray*}
	where $\sup_{1 \le k \le p}(n^{-1}\bx_{(k)}\bx_{(k)}^{T}) \le \kappa_2 (1 + \eta)$ is used in the above inequality. Hence, it follows from Lemmas 5 and 6, together with Condition 3, that $\pr\left(\mathcal{D}_n^c\right) \to 0$ as $n \to \infty$. This completes the proof.
	
	\section{Proof of Theorem 2}
		Since the autoregressive model with order $d$ can be formulated as a autoregressive model with order 1, without loss of generality, we consider the case of order 1 only. With probability tending to one, $\hat{k} = k_0$, and thus it suffices to consider the set
		$
		\mathcal{A}_n = \{ \hat{k} = k_0\}.
		$
		Over the set $\mathcal{A}_n$, for each $i$,
		\be
		\label{th2-proof-lse}
		\wh{\bbeta}_{i} - \bbeta_{i} = \left(\bX_{i}^{\T}\bX_{i}\right)^{-1}\bX_{i}^{\T}\bbe_{(i)}.
		\ee
		For each $i$, the law of large numbers for the stationary process case yields that
		$
		n^{-1}\bX_{i}^{\T}\bX_{i}
		$
		converges to a positive matrix almost surely, and furthermore, with probability tending to one, $\lambda_{\min}\left(n^{-1}\bX_{i}^{\T}\bX_{i} \right) $ is bounded away from zero. As a matter of fact, if we define
		$$
		\mathcal{B}_n = \bigcap_{{1\le i \le p}} \left\{\lambda_{\min}\left(n^{-1}\bX_{i}^{\T}\bX_{i} \right) > \kappa_1 (1 -\eta)
		\right\}
		$$
		with a small constant $\eta \in (0,1)$, then it follows from by Lemma 5 or Lemma 6 under different moment conditions that $P\{\mathcal{B}_n\} \to 1$ as $n \to \infty$.
		Hence, over the event $\mathcal{A}_n \cap \mathcal{B}_n$,
		$$
		\left\|\wh{\bbeta}_{i} - \bbeta_{i}\right\|_2^2 \le \kappa_1^{-2}(1-\eta)^{-2} n^{-2}\left\|\bbe_{(i)}^{\T}\bX_{i}\right\|^2_2 = C_1 n^{-2}\left\|\bbe_{(i)}^{\T}\bX_{i}\right\|^2_2,
		$$
		where $C_1 = \kappa_1^{-2}(1-\eta)^{-2} > 0$.
		It is not hard to see from Lemma 5(ii) or Lemma 6(ii) that, for all $1 \le i \le p$, $n^{-1}E\|\bX_{i}^{\T}\bbe_{(i)}\|^2_2 \le C_2 $ with some constant $C_2 > 0$. Therefore,
		for a large positive constant $C$, we obtain that
		\begin{eqnarray*}
			\pr\left(\left \|\widehat{\bA}_{1} - \bA_1\right\|_F^2 > C {n^{-1}p}\right)
			&=& \pr\left(\left \|\widehat{\bA}_{1} - \bA_1\right\|_F^2 > C {n^{-1}p}, \mathcal{A}_n\cap \mathcal{B}_n\right) + \pr\left(\mathcal{A}_n\cap \mathcal{B}_n\right) \\
			&=& (C p)^{-1} n  (C_1 n^{-2}) E \left(\sum_{i = 1}^p \|\bX_{i}^{\T}\bbe_{(i)}\|^2_2\right)  + \pr\left(\mathcal{A}_n\cap \mathcal{B}_n\right) \\
			&=& {C_1 C_2 C^{-1}}  + o(1).
		\end{eqnarray*}
		We establish the convergence rate of $\|\widehat{\bA}_1 - \bA_1\|_F$ by taking a sufficiently large $C$.
		
		Now we derive the convergence rate of $\|\widehat{\bA}_1 - \bA_1\|_2$. For any matrix $\bB$, $\|\bB\|_2^2 \le \|\bB\|_1 \|\bB\|_{\infty}$. Hence, on the event $\mathcal{A}_n$,
		$$
		\|\widehat{\bA}_1 - \bA_1\|_2 \le \surd{\|\widehat{\bA}_1 - \bA_1\|_1} \surd{\|\widehat{\bA}_1 - \bA_1\|_\infty} \le (2 k_0+1)\sup_{i \le p, j \le \tau_i} \big|\wh{\beta}_{ij} - \beta_{ij}\big|,
		$$
		where $\wh{\beta}_{ij}$ and $\beta_{ij}$ are the $j$-th element of $\wh{\bbeta}_i$ and $\bbeta_i$, respectively.
		Observe from (\ref{th2-proof-lse}) that
		$$
		\sup_{i \le p, j \le \tau_i} \big|\wh{\beta}_{ij} - \beta_{ij}\big| = \kappa_1^{-1} (1-\eta)^{-1} (2k_0 + 1) \big(\sup_{i \le p, j \le \tau_i} |\bbe_{(i)}^{\T} \bx_{(j)}|\big), i = 1,\ldots,p.
		$$
		Hence, using Lemma 5(ii) or Lemma 6(ii), we have
		$$
		\sup_{i \le p, j \le \tau_i} \big|\wh{\beta}_{ij} - \beta_{ij}\big| = O_P\left\{ \left(n^{-1}\log p \right )^{1/2}\right\},
		$$
		which shows that
		$$
		\|\widehat{\bA}_1 - \bA_1\|_2 = O_P\left\{ \left(n^{-1}\log p \right )^{1/2}\right\}.
		$$
		The proof is completed.
		
		\section{Proof of Theorem 3}
			The covariance matrix $\bSigma_0$ can be expressed as
			$$
			\bSigma_0 = \bSigma_{\bep} + \sum_{j=1}^\infty \bB_j,~~ \bB_j =
			\bJ \widetilde{\bA}^j \bJ^{\T} \bSigma_{\bep} \bJ(\widetilde{\bA}^{\T})^j\bJ^{\T}, ~~j \ge 1,
			$$
			where $\bJ= \left(\bI_{p \times p}, \textit{0}_{p \times (d-1)p}\right)$. Let $\bPhi_j = \bJ \widetilde{\bA}^j \bJ^{\T}$, $j \ge 1$. By the companion matrix $\widetilde{\bA}$, we can show that $\bPhi_0 = \bI_p$ and $\bPhi_j = \sum_{k=1}^{\min(j,d)} \bPhi_{j - k}\bA_k$, $j \ge 1$. It is easy to see that for two banded matrices $\bF$ and $G$ with bandwidths $2r_1+1$ and $2r_2+1$, respectively, the product matrix $\bF G$ is also banded and its bandwidth is at most $2(r_1+r_2)+1$. Therefore, it can be verified that $\bPhi_j$ is banded with bandwidth at most $2jk_0+1$ and then $\bB_{j}$ is also banded with its bandwidth at most $2(2jk_0 + s_0)+1$ for $j \ge 1$. Take $\bSigma_0^{(r)} = \bSigma_{\bep} + \sum_{j=1}^r \bB_j $, which is banded with the bandwidth at most $2(2rk_0+s_0)+1$, and $\bSigma_0 - \bSigma_0^{(r)} = \sum_{j= r+1}^\infty \bB_j$. Note that for any $j\ge 1$, $\|\bB_j\|_2 \le \|\bSigma_{\bep}\|_2 \|\widetilde{\bA}^{2j}\|_2 \le C \delta^{2j}$ for some $C>0$. Write $C_1 = C \|\bSigma_{\bep}\|_2\left(1-\delta^2\right)^{-1}$. It follows that
			\bs
			\|\bSigma_0 - \bSigma_0^{(r)}\|_2 \le \sum_{j= r+1}^\infty \|\bB_j\|_2 \le C \|\bSigma_{\bep}\|_2\left(1-\delta^2\right)^{-1} \delta^{2(r+1)} = C_1\delta^{2(r+1)}.
			\es
			By using the inequality $\|\bB_j\|_1 \le \{2(2jk_0 + s_0)+1\} \|\bB_j\|_2 \le C (2j+1) \delta^{2j}$ for some $C>0$, we obtain
			$$
			\|\bSigma_0 - \bSigma_0^{(r)}\|_1 \le C_2 r \delta^{2(r+1)}.
			$$
			Other inequalities can be proved analogously. The proof is complete.

\section{Proof of Theorem 4}
	Now we prove the convergence rate of $\|\widehat{\bSigma}_{n,0}^{(r_n)} - \bSigma_0\|_2$. First,
	$\|\widehat{\bSigma}_{n,0}^{(r_n)} - \bSigma_0\|_2$ can be bounded above by
	$$
	\|\widehat{\bSigma}_{n,0}^{(r_n)} - \bSigma_0^{(r_n)}\|_2 + \|\bSigma_0^{(r_n)} - \bSigma_0\|_2 = R_{n1} + R_{n2}.
	$$
	Similar to Theorem 2,
	$
	R_{n1} \le (4 r_n k_0+2 s_0+1) \sup_{j,k\le p }|\wh{\Sigma}_{jk} - \Sigma_{jk}|$. From Lemma 5(i) or Lemma 6(i), we obtain that
	\begin{eqnarray*}
		R_{n1} = O_P \left\{r_n \left( n^{-1} \log p \right)^{1/2} \right\}.
	\end{eqnarray*}
	From Theorem 3, $R_{n2} \le O(\delta^{2(r_n+1)})$. Note that $r_n= C \log \{n \log^{-1}(p)\}$ with $C > \left( - 4 \log \delta \right)^{-1}$. Combining these results, it follows that
	\begin{eqnarray*}
		\|\widehat{\bSigma}_{n,0}^{(r_n)} - \bSigma_0\|_2 &=& O_P\Big\{ r_n\left(n^{-1}\log p \right )^{1/2} +  \delta^{2(r_n+1)} \Big\} \\
		&=&		O_P\Big[\log \{n \log^{-1}(p)\}\left(n^{-1}\log p \right )^{1/2} \Big].
	\end{eqnarray*}
	The proofs of other results are similar and omitted.
	
	\section{Proposition 1 and its proof}
	
	{PROPOSITION 1.} {\it Under Conditions 1-4 in Section 3.1 of the original article, we prove that $\pr (\hat{k} = k_0, \hat{d} = d) \to 1$ as $n, p \to \infty$.}
	
 \noindent	{\bf Proof of Proposition 1.} 
		Our primary goal is to prove that
		$ \pr(\hat{k}= k_0,\hat{d} = d) \to 1$, i.e., $$ \pr\{(\hat{k} \neq k_0) \cup (\hat{d} \neq d)\} \to 0.$$
		Note that
		$$
		\pr\{(\hat{k} \neq k_0) \cup (\hat{d} \neq d)\} \le \pr(\hat{k} < k_0) + \pr( \hat{d} < d) +
		\pr(\hat{k} > k_0,\hat{d} > d).
		$$
		We observe that both events $\{\hat{k} < k_0\}$ and $\{\hat{d} < d\}$ correspond to the underfitting case, where some important variables are missed in the estimated model. Hence, following the proofs of Theorem 1, we can show
		$
		\pr(\hat{k} < k_0) + \pr( \hat{d} < d) \to 0.
		$
		
		It remains to prove that $ \pr(\hat{k} > k_0,\hat{d} > d) \to 0$. First look at the event $\mathcal{A} = \big\{\hat{k} > k_0,\hat{d} > d\big\}$. Define
		$\mathcal{A}_1 = \underset{ i \le p }{\cup} \big\{\hat{k}_i \ge k_0, \hat{d}_i > d\big\}$,
		$\mathcal{A}_2 = \underset{ i \le p }{\cup} \big\{\hat{k}_i < k_0, \hat{d}_i >  d\big\} $,
		$\mathcal{A}_3 = \underset{ i \le p }{\cup} \big\{\hat{k}_i > k_0, \hat{d}_i \ge d\big\} $,
		and
		$\mathcal{A}_4 = \underset{ i \le p }{\cup} \big\{\hat{k}_i > k_0, \hat{d}_i < d \big\} $.
		Then 
		$
		\mathcal{A} \subset \mathcal{A}_1\cup\mathcal{A}_2\cup\mathcal{A}_3\cup\mathcal{A}_4,
		$
		which implies that it suffices to show $\pr(\mathcal{A}_k) \to 0$ for each $k = 1,\ldots,4$.
		Observe that both events $\mathcal{A}_1$ and $\mathcal{A}_3$ correspond to the overfitting case, where all important variables as well as some unimportant variables are selected by the estimated model. Hence, following the proofs of Theorem 1, we can show
		$
		\pr(\mathcal{A}_1) + \pr( \mathcal{A}_3) \to 0.
		$
		
		Now we are going to prove $\pr(\mathcal{A}_2) \to 0$ as $n,p \to \infty$.  For each $i$, $\{\hat{k}_i < k_0, \hat{d}_i > d\}$ means
		$
		\underset{{k < k_0, \ell > d} }{\min}\widetilde{\bic}_{i}(k,\ell) < \widetilde{\bic}_{i}(k_0,d).
		$
		Hence, we only need to show, with probability tending to one,
		\begin{eqnarray}
		\label{gbic-1}
		\min_{i \le p }\min_{k < k_0, d < \ell \le L} \left\{\widetilde{\bic}_{i}(k,\ell) - \widetilde{\bic}_{i}(k_0,d) \right \} > 0.
		\end{eqnarray}		
		Suppose that we have shown that there exists a constant $\eta >0$ and an event $\mathcal{G}_n$ such that
		$
		\pr(\mathcal{G}_n) \to 1
		$
		as $n \to \infty$ and on the event $\mathcal{G}_n$,
		\begin{eqnarray}
		\label{append-3}
		\min_{i \le p} \big \{\rss_i(k,\ell) - \rss_i(k_0,d) - \eta \rss_{i}(k_0,d) \Delta_i\ge 0,
		\end{eqnarray}
		for each $ k < k_0$, $d < \ell < L$ and sufficiently large $n$, where
		$
		\Delta_i = \sum_{j = 1}^d \big\{(a_{i,i-k_0}^{(j)})^2 + (a_{i,i+k_0}^{(j)})^2\big\}.
		$ As a result, on the event $\mathcal{G}_n$ with large $n$,
		$
		\min_{i \le p } \big \{\log \rss_i(k,\ell) - \log \rss_i(k_0,d)
		- \log(1 + \eta \Delta_i) \big\} \ge 0.
		$
		Note that $\log (1 + x) \ge \min\{0.5x,\log 2\}$ for any $x > 0$.
		Then, with probability tending to one,
		$\log \rss_i(k) - \log \rss_i(k_0)$ can be further bounded below by
		$
		\min(0.5 \Delta_i,\log 2).
		$
		Condition 3 implies that $\min_{i \le p} \Delta_i\gg C_n n^{-1} \log (p \vee n)$ as $n \to \infty$.  Hence, it follows that, with probability tending to 1,
		\begin{eqnarray}
		\widetilde{\bic}_{i}(k,\ell) -  \widetilde{\bic}_{i}(k_0,d)
		\ge \min(0.5 \eta\Delta_i,\log 2) - C_n\tau_{i}(k_0,d) n^{-1} \log (p \vee n) > 0
		\end{eqnarray}
		uniformly for all $k < k_0, d < \ell \le L$ and $i = 1,\ldots,p$. Hence,
		$\pr(\mathcal{A}_2) \to 0$ as $n,p \to \infty$.
		
		Let us turn to prove (\ref{append-3}). For $k < k_0$ and $d < \ell \le L$, denote $\bH_{i,k,\ell} = \bX_{i,k,\ell}\big(\bX_{i,k,\ell}^{\T}\bX_{i,k,\ell}\big)^{-1}\bX_{i,k,\ell}^{\T}$, where $\bX_{i,k,\ell}$ is defined as in section 2.2 but replaced $k_0$ and $d$ by $k$ and $\ell$. Then
		$
		\rss_i(k,\ell)= \by_{(i)}^{\T}\big(\bI_{n-1} - \bH_{i,k,\ell}\big)\by_{(i)}.
		$
		In fact,  $\bX_{i,k_0,\ell}$ can be rewritten as
		$\bX_{i,k_0,\ell} = (\bS_{i,k,1}^{(1)},\bX_{i,k}^{(1)}, \bS_{i,k,2}^{(2)},\ldots,\bS_{i,k,1}^{(\ell)},\bX_{i,k}^{(\ell)}, \bS_{i,k,2}^{(\ell)})$
		and, similarly,
		$
		\bbeta_{i,k_0,\ell}^T = \big(\bb_{i,1}^{(1)},\bbeta_{i,k}^{(1)}, \bb_{i,2}^{(1)}, \ldots,\bb_{i,1}^{(\ell)},\bbeta_{i,k}^{(\ell)}, \bb_{i,2}^{(\ell)}\big).
		$
		Let
		$
		\bS_{i,\ell} = (\bS_{i,k,1}^{(1)},\bS_{i,k,2}^{(2)},\ldots,\bS_{i,k,1}^{(\ell)},\bS_{i,k,2}^{(\ell)})
		$
		and
		$
		\bb_{i,\ell}^T = \big(\bb_{i,1}^{(1)},\bb_{i,2}^{(1)}, \ldots,\bb_{i,1}^{(\ell)},\bb_{i,2}^{(\ell)}\big).
		$
		As a result, by Lemma 5 (ii) or Lemma 6 (ii), we have
		$$
		\max_{i \le p }\big|\rss_i(k,\ell) - \rss_i(k_0,d) - \bb_{i,\ell}^{\T}\bS_{i,\ell}^{\T}(\bI_{n-1} - \bH_{i,k,\ell})\bS_{i,\ell} \bb_{i,\ell}\big| = o_P(1).
		$$
		From Lemma 5 (i) or Lemma 6 (i) and Lemma 7, there exists a small constant $ \eta >0$ such that, with probability tending to one,
		$$
		\lambda_{\min}\big\{\bS_{i,\ell}^{\T}(\bI - \bH_{i,k,\ell}) \bS_{i,\ell}\big\} > \eta (1 +  \eta) n\sigma_i^2,
		$$
		and
		$\rss_i(k_0,d) \le n\sigma_i^2 (1 + \eta)$. Note that
		$
		\bb_{i,\ell}^T\bb_{i,\ell} \ge \sum_{j = 1}^d \big\{(a_{i,i-k_0}^{(j)})^2 + (a_{i,i+k_0}^{(j)})^2\big\}.
		$
		Therefore, (\ref{append-3}) follows.
		
		In a similar manner, $\pr(\mathcal{A}_4) \to 0$ can be proved. The proof is completed.

		\section{Proposition 2 and its proof}
		
		{PROPOSITION 2.} {\it Under Conditions 1' and 2--4 in Section 3.1 of the original article, $\pr(\hat{k} = k_0) \to 1$ as $n, p \to \infty$,
			provided $ k_0 \ll C_n^{-1}n/\log(p\vee n)$.	}

  \noindent {\bf Proof of Proposition 2.} 
			First, we can prove the conclusions of Lemma 5, 6 and 7 under the Conditions 1' and (2)--(4). For instance, in the proof of Lemma 5, we bound $\|A_1^{l}\|_{\infty} $ in (\ref{Ainfty}) by $\|A_1^{l}\|_{\infty} \le C \delta^l$ under Condition 1'. Similarly, the inequalities (\ref{lemma6_1}) and (\ref{lemma6_2}) in the proof of Lemma 6 can be bounded in a similar way.  Then, following the proof of Theorem 1, we can prove the consistency of Bayesian information criterion selector $\hat{k}$ in the general setting $k_0 \to \infty$.
			
\section{Seven technical lemmas and their proofs}

We first adopt the asymptotic theories using the functional dependent  measure of \citet{W2005}. Assume that $z_i$ is a stationary process of the form $z_i = g(\mathcal{F}_i)$, where $g(\cdot)$ is a measurable function and $\mathcal{F}_i = (\ldots,e_{-1},e_0,\ldots,e_i)$ with independent and identically distributed random variables $\{e_i; i = 0,\pm 1,\ldots\}$. \citet{W2005} defined the functional dependent measure in terms of how the outputs are affected by the inputs. To be specific, denote $\|z\|_q = \left\{E(|z|^q)\right\}^{1/q}$ with $q \ge 1$ for a random variable $z$. The physical or functional dependent measure is defined as
$$
\theta_{i,q} = \|z_i - z_i^*\|_q = \|g(\mathcal{F}_i) - g(\mathcal{F}_i^*)\|_q,
$$
where $z_i^* = g(\mathcal{F}_i^*)$ is the coupled process of $z_i$, $\mathcal{F}_i^* = (\ldots,e_{-1},e_0^*,\ldots,e_i)$ with $\{e_0^*,e_0\}$ being independent and identically distributed. Intuitively, $\theta_{i,q}$ measures the dependency of $z_i$ on $e_0$ while keeping all other innovations unchanged.

\begin{lemma}
	\label{lem1}
	{\it (Theorem 2 (ii) of \citet{LXW2013}).} Let $S_n = n^{-1/2}\sum_{i = 1}^n z_i$ and $\Theta_{m,q} = \sum_{i = m}^\infty \theta_{i,q}$. Assume that for each $m$, $\Theta_{m,q} = O(m^{-\alpha})$ with $\alpha > 1/2 - 1/q$ and $q > 2$. Then there exist positive constants $C_1$, $C_2$ and $C_3$ which only depend on $q$ such that for all $x > 0$,
	$$
	\pr\big( |S_n| \ge x \big) \le \frac{C_1 \Theta_{0,q}^q n} {(n^{1/2} x)^q} + C_3 \exp\big(C_2 \Theta_{0,q}^{-1} x^2\big).
	$$
	
\end{lemma}
		
	To prove the limit theory for the sub-exponential tail case under Condition 4(ii), we shall use Lemmas 2--4.
	
	\begin{lemma}
		\label{lem2}
		Suppose that $X$ is a random variable. Then, $E\left\{\exp(t_0 |X|^v)\right\} < \infty$ for some $ 0< v \le 2$ and $t_0 > 0 $ if and only if
		$$
		\lim \sup_{q \to \infty} q^{- 1/v} \|X\|_q < \infty.
		$$
	\end{lemma}
	\begin{proof}
		Assume that $\zeta = E\{\exp(t_0 |X|^v)\} < \infty$. Then, for any $q \ge 2$,
		\begin{eqnarray*}
			E(|X|^{q}) &=& q \int_0^\infty x^{q-1} \pr (|X| > x)d x \\
			& \le &  \zeta q v^{-1} t_0^{ - q/v } \int_0^\infty x^{q/v -1} \exp\big( - x\big)d x
			=   \zeta q v^{-1} t_0^{ - q/v } \Gamma\left(\frac{q} {v}\right),
		\end{eqnarray*}
		where $\Gamma(\cdot)$ is the Gamma function. By Stirling's formula,
		$$
		\lim_{x \to \infty} \Gamma(x+1) \Big\{\ (2 \pi x)^{1/2} \left(\frac{x}{e}\right)^{x}\Big\}^{-1} = 1,
		$$
		we obtain that for all sufficiently large $q$,
		$$
		\|X\|_q \le \left (\zeta q v^{-1} t_0^{ - q/v} \right)^{1/q} q^{1/(2q)} \left(\frac{q} {v} - 1\right)^{1/v - 1 /q} \le C q^{1/v},
		$$
		where $C$ is a constant depending on $\zeta$, $v$ and $t_0$ only. This implies that
		$$
		\lim \sup_{q \to \infty} q^{- 1/v} \|X\|_q < \infty.
		$$
		
		Conversely, assume that $ \lim \sup_{q \to \infty} q^{- 1/v} \|X\|_q < \infty.$ Then, there exists a positive constant $\phi_0 > 0$ such that,
		$
		\|X\|_q \le \phi_0  q^{1/v}
		$
		for all $q \ge 2$. Note that $\exp(x) = 1 + \sum_{k \ge 1} (k!)^{-1} x^k.$ To prove that $E\{\exp(t_0 |X|^v)\} < \infty$ for some $t_0 > 0$, we only need to show that there exist positive constants $t_0 $ and $k_0$ such that
		$$
		\sum_{k \ge k_0}  {t_0^k\|X\|_{v k}^{v k} \over k!} < \infty.
		$$
		By Stirling's formula, there exists a large integer $k_0$ such that for $k \ge k_0$,
		$$
		\Gamma(k+1) = k! \ge  ( \pi k)^{1/2} \left({k \over e}\right)^{k}.
		$$
		With such $k_0$ and $t_0 = (2 \phi_0^v v e)^{-1}$, we have
		$$
		\sum_{k \ge k_0}  {t_0^k\|X\|_{v k}^{v k} \over k!}
		\le \sum_{k \ge k_0}  {(t_0\phi_0^v v e )^k k^k \over  (\pi k)^{1/2} k^k}
		\le \sum_{k \ge k_0}  {2^{-k}  } < \infty.
		$$
		\end{proof}
		
		\begin{lemma}
			\label{lem3}
			Suppose that $\{X_1,\ldots,X_n\}$ are independent random variables and $\sup_{i \le n}E\{\exp(t_0 |X_i|^\alpha)\} \le \zeta $ for some positive constants $\alpha$, $t_0$ and $\zeta$ with $0 < \alpha \le 1$. Then there exist positive constants $C_j>0$ $(j =1,\ldots,4)$ which depend only on $\alpha$, $t_0$ and $\zeta$ such that for any $x > 0$ and all $n$, the following concentration inequality holds:
			\begin{eqnarray}
			\label{lemma3}
			\pr\Big[\Big|\sum_{i=1}^n \{X_i - E(X_i)\} \Big| > 3 x \Big]
			&\le & C_1\exp\left( - { x^2 \over   C_2 n + C_3 n^{1-\alpha \over 2-\alpha} x}\right) \nonumber\\
			&&  + C_1 \exp\left( - { x^{2\alpha} \over C_2 n^{\alpha \over 2-\alpha} + C_3 x^{\alpha}}\right)
			+ n C_1 \exp\left(- C_4 x^{\alpha}\right).
			\end{eqnarray}
			In particular, if $\alpha = 1$, then
			$$
			\pr\Big[\Big|\sum_{i=1}^n \{X_i - E(X_i)\} \Big| > 3 x \Big]
			\le C_1 \exp\left(- { x^2 \over  C_2 n + C_3 x}\right) + C_1 n \exp\left(- C_4 x\right)
			$$
			for any $ x > 0$ and $n$. 
			
		\end{lemma}
		\begin{proof}
		For the case of $\alpha=1$, (\ref{lemma3}) can be proved by Bernstein's inequality directly. So here we consider the case of $0 <\alpha < 1$ only. Let $\xi_{n1}$ and $\xi_{n2}$ be two constants with $0< \xi_{n1} < \xi_{n2}$, which depend on $n$ and will be defined below. Let $\tilde{X}_{i1} = X_i I(|X_i| \le \xi_{n1})$, $\tilde{X}_{i2} = X_i I(\xi_{n1}\le |X_i| \le \xi_{n2})$ and $\tilde{X}_{i3} = X_i I( |X_i| > \xi_{n2})$. Then
			$
			X_i = \tilde{X}_{i1} - E(\tilde{X}_{i1}) + \tilde{X}_{i2} - E(\tilde{X}_{i2}) +
			\tilde{X}_{i3} - E(\tilde{X}_{i3}),
			$
			and hence
			$$
			\pr\Big[\Big|\sum_{i=1}^n \{X_i - E(X_i)\} \Big| > 3 x \Big] \le \sum_{k=1}^3 \pr\Big[\Big|\sum_{i=1}^n\{\tilde{X}_{ik} - E(\tilde{X}_{ik})\} \Big| > x \Big].
			$$
			In the following, we will give an upper bound on each term separately.
			
			Now consider the first term. Let $\sigma^2$ be a finite constant  such that $\sup_{i \le n} E|X_i|^2 \le \sigma^2$.  Note that $|\tilde{X}_{i1}|\le \xi_{n1}$ and $E\tilde{X}_{i1}^2 \le \sigma^2 $ for all $i$. By Bernstein's inequality for bounded variables, we get that
			\begin{eqnarray}
			\label{lemma3-1}
			\pr\Big[\big|\sum_{i=1}^n\big\{\tilde{X}_{i1} - E(\tilde{X}_{i1})\big\}\big| > x\Big]
			&\le& 2 \exp\Big( - { x^2 \over 2n  \sigma^2 + 2\xi_{n1} x/3}\Big).
			\end{eqnarray}
			
			Let us handle the second term. To use Bernstein's equality, we only require an appropriate control of moments.
			Using integration by parts, we observe that
			$$
			E\big(|\tilde{X}_{i2}|^q\big) \le q \int_{\xi_{n1}}^{\xi_{n2}} u^{q-1} \pr(|X_i| > u)d u + \xi_{n 1}^q \pr(|X_i| > \xi_{n1})
			$$
			for $q \ge 2$. For integer $q \ge 2$,
			\begin{eqnarray*}
				q \int_{\xi_{n1}}^{\xi_{n2}} u^{q-1} \pr(|X_i| > u)du
				&\le &  q \zeta \int_{\xi_{n1}}^{\xi_{n2}} u^{q-1} \exp(- t_0 u^\alpha) d u \\
				&\le &  q  \alpha^{-1} \zeta (2 t_0^{-1})^{q/\alpha}\int_{t_0 \xi_{n1}^{\alpha}/2}^{t_0 \xi_{n2}^{\alpha}/2} u^{q/\alpha -1} \exp(- 2 u) d u \\
				&\le &  q \alpha^{-1}\zeta  \big(2 t_0^{-1} \xi_{n2}^{1-\alpha}\big)^{q} \exp(- 2^{-1} t_0 \xi_{n1}^\alpha ) \int_{t_0 \xi_{n1}^{\alpha}/2}^{t_0 \xi_{n2}^{\alpha}/2} u^{q -1} \exp(- u) du \\
				&\le & q!  4\alpha^{-1} \zeta \big(t_0^{-1} \xi_{n2}^{1-\alpha}\big)^2 \exp(- 2^{-1} t_0 \xi_{n1}^\alpha) \big(2t_0^{-1} \xi_{n2}^{1-\alpha}\big)^{q-2}.
			\end{eqnarray*}
			Choose $\xi_{n1} = \{ 4t_0^{-1} {(1- \alpha)/(2-\alpha)} \log n\}^{1/\alpha}$ and $\xi_{n2} = n^{ 1 /(2-\alpha)} \vee x$. Write $\xi_n = n^{{(1-\alpha)/ (2-\alpha)}}$ and $\nu = \max(16 \zeta \alpha^{-1} t_0^{-2}, \sigma^2)$. Then
			\begin{eqnarray*}
				q \int_{\xi_{n1}}^{\xi_{n2}} u^{q-1} \pr(|X_i| > u)du
				\le  {1\over 2}q!  \nu \big\{ 1 \vee x^{2(1-\alpha )} \xi_n^{-2}\big\} \big\{2t_0^{-1} (\xi_{n} \vee x^{1-\alpha})\big\}^{q-2}.
			\end{eqnarray*}
			We also have that
			$
			\xi_{n1}^q \pr(|X_i| >  \xi_{n1}) \le \xi_{n1}^q \exp(- t_0 \xi_{n1}^\alpha) = \xi_{n1}^{2} \exp(- t_0 \xi_{n1}^\alpha) \xi_{n1}^{q-2}.
			$
			A simple manipulation yields that there exists a positive integer $N_{\alpha,t_0}$ which depends only on $\alpha$ and $t_0$ such that
			$$
			\xi_{n1} < \xi_{n2}, ~~\xi_{n1}^2 \exp(- t_0 \xi_{n1}^\alpha) \le 4 \alpha^{-1} \zeta t_0^{-2}, ~~2 t_0^{-1} \xi_{n2}^{1-\alpha} \ge \xi_{n1},~~ \mbox{and} ~~4 \log n \le t_0 \xi_{n2}^\alpha,
			$$
			if $n > N_{\alpha,t_0}$. Then, if $x \le n^{1/( 2-\alpha)}$,
			$$
			E(|\tilde{X}_{i2}|^q) \le {1 \over 2} q! \nu \left(2t_0^{-1}\xi_n \right)^{q-2}
			$$ for $q \ge 2$;  otherwise,
			$$
			E(|\tilde{X}_{i2}|^q) \le {1 \over 2} q! \nu \big\{x^{2(1-\alpha )} \xi_n^{-2}\big\} \big(2t_0^{-1} x^{1-\alpha}\big)^{q-2}
			$$ for $q \ge 2$. By Bernstein's inequality, we obtain that
			\begin{eqnarray}
			\label{lemma3-2}
			\pr\Big[\big|\sum_{i=1}^n\big\{\tilde{X}_{i2} - E(\tilde{X}_{i2})\big\}\big| > x\Big]
			&\le &2 \exp\Big( - { x^2 \over 2 n  \nu + 4t_0^{-1} n^{1-\alpha \over 2-\alpha} x}\Big) \nonumber \\
			&&+ 2 \exp\Big( - { x^{2\alpha} \over 2  \nu n^{\alpha \over 2-\alpha} + 4t_0^{-1}  x^{\alpha}}\Big).
			\end{eqnarray}
			
			For the last term, we note that
			\begin{eqnarray*}
				\Big[|\sum_{i=1}^n \{\tilde{X}_{i3} - E(\tilde{X}_{i3})\}| > x \Big] &\subset& \Big\{\sup_{i}|X_i| > \xi_{n2} \Big\} \cup \Big\{\sup_{i}|X_i| \le \xi_{n2},\\
				&&  \hskip 0.5cm \sum_{i=1}^n |E(X_i I(|X_i| > \xi_{n2}))| > x \Big\}.
			\end{eqnarray*}
			Therefore, we have
			\begin{eqnarray*}
				&&\pr\Big[\Big|\sum_{i=1}^n (\tilde{X}_{i3} - E(\tilde{X}_{i3})\}\Big| > x \Big]\\
				&&\le \pr\Big(\sup_{i}|X_i| > \xi_{n2} \Big) + \pr\Big[\sup_{i}|X_i| \le \xi_{n2}, \sum_{i=1}^n |E\{X_i I(|X_i| > \xi_{n2})\}| > x \Big].
			\end{eqnarray*}
			Note that $ \zeta = \sup_{i \le n} E\{\exp(t_0 |X_i|^\alpha)\} < \infty$. We observe that
			\begin{eqnarray*}
				\pr\Big(\sup_{i}|X_i| > \xi_{n2} \Big) \le \zeta n \exp\Big( - t_0 \xi_{n2}^\alpha\Big) \le \zeta n \exp\Big( - t_0 x^\alpha\Big).
			\end{eqnarray*}
			In a similar fashion, we obtain that
			$$
			\sum_{i=1}^n |E\{X_i I(|X_i| > \xi_{n 2})\}| \le n\sigma \left\{\pr(|X_i| > \xi_{n 2})\right\}^{1/2} \le  n^{-1}\sigma \zeta \exp\Big( 2 \log n - 2^{-1}t_0\xi_{n 2}^\alpha\Big).
			$$
			As a result, for $x > \sigma \zeta n^{-1}$ and $n > N_{\alpha,t_0}$,
			\begin{eqnarray}
			\label{lemma3-3}
			\pr\Big[|\sum_{i=1}^n \{\tilde{X}_{i3} - E(\tilde{X}_{i3})\}| > x \Big] \le \zeta n \exp\Big( - t_0 x^\alpha\Big).
			\end{eqnarray}
			Combing the three inequalities (\ref{lemma3-1})-(\ref{lemma3-3}), we conclude that, for $x > \sigma \zeta n^{-1}$ and $n > N_{\alpha,t_0}$,
			\begin{eqnarray*}
				\pr\Big [\Big|\sum_{i=1}^n \{X_i - E(X_i)\}\Big| > 3 x \Big]
				&\le & 4\exp\left( - { x^2 \over 2 n  \nu + 4t_0^{-1} n^{1-\alpha \over 2-\alpha} x}\right)\\
				&&  + 2 \exp\left( - { x^{2\alpha} \over 2  \nu n^{\alpha \over 2-\alpha} + 4t_0^{-1}  x^{\alpha}}\right)
				+ n \zeta \exp\left(-t_0 x^{\alpha}\right).
			\end{eqnarray*}
			If $x \le \sigma \zeta n^{-1}$ or $n \le N_{\alpha,t_0}$, we can always multiply a large positive constant $C$ on the right hand side to make the inequality hold. The proof is completed.
		\end{proof}
		\begin{lemma}
			\label{lem4}
			Suppose that $\{\bX_1 = (X_{1,1},X_{1,2})^\top,\bX_2 = (X_{2,1},X_{2,2})^\top,\ldots\}$ are independent random vectors and $\sup_{i \le n, j = 1,2} E\{\exp(t_0 |X_{i,j}|^{2\alpha})\} \le \zeta$ for some positive constants $\alpha$, $t_0$ and $\zeta$ with $0 <\alpha \le 1$. Denote by $l_n$ a sequence that may depend on $n$, and $1 \le l_n \le O(n^\epsilon)$ with $0 \le \epsilon < 1$. Then, for each $m$ and $m'$ with $m,m' = 1,2$, there exist positive constants $C_j (j = 1,\ldots,4)$ such that for any $x > 0$, the following concentration inequality holds:
			\begin{eqnarray*}
				&& \pr\Big[\Big|\sum_{i=1}^n \{X_{i,m} X_{i+l_n,m'} - E(X_{i,m} X_{i+l_n,m'})\} \Big|
				> 3 (l_n+1) x \Big] \\
				& \le &  (l_n+1) C_1\exp\left( - { x^2 \over C_2 n  + C_3 n^{1-\alpha \over 2-\alpha} x}\right) + C_1 (l_n + 1) \exp\left( - { x^{2\alpha} \over C_2 n^{\alpha \over 2-\alpha} + C_3 x^{\alpha}}\right)\\
				& &   + C_1 (l_n + 1) n \exp\left(-C_4 x^{\alpha}\right).
			\end{eqnarray*}
		\end{lemma}
		\begin{proof} 
			Without loss of generality, we assume that $n/(l_n+1)$ is a positive integer. Here we prove the inequality for $m =1$ and $m'=2$ only. Similar techniques can be applied to other cases. Let $Y_{ji} = X_{(i - 1)(l_n+1) + j,1} X_{i (l_n+1) + j - 1,2}$. Then, for each $j$, $\{Y_{ji}, i = 1,\ldots,n/(l_n+1)\}$ are independent with $\sup_{i,j} E\{\exp(t_0 |Y_{ji}|^{\alpha})\} \le \zeta < \infty$. With the help of $Y_{ji}$,
			$\sum_{i = 1}^n \big\{ X_{i,1} X_{i+l_n,2} - E(X_{i,1} X_{i+l_n,2})\big\}$ can be re-expressed as
			\begin{eqnarray*}
				\sum_{i = 1}^n \big\{ X_{i,1} X_{i+l_n,2} - E(X_{i,1} X_{i+l_n,2})\big\} = \sum_{j = 1}^{l_n+1} \sum_{i = 1}^{n/(l_n+1)} \big\{Y_{ji} - E(Y_{ji})\big\}.
			\end{eqnarray*}
			By Lemma 3, we obtain that there exist positive constants $C_j (j = 1,\ldots,4)$ such that
			\begin{eqnarray*}
				\pr\Big[\Big|\sum_{i = 1}^{n/(l_n+1)} \big\{Y_{ji} - E(Y_{ji})\big\}\Big| > 3 x\Big]
				&\le& ~ C_1\exp\left( - { x^2 \over C_2 n  + C_3 n^{1-\alpha \over 2-\alpha} x}\right) \\
				&& + ~ C_1 \exp\left( - { x^{2\alpha} \over C_2 n^{\alpha \over 2-\alpha} + C_3 x^{\alpha}}\right) \\
				&& + ~C_1 n \exp\left(-C_4 x^{\alpha}\right),
			\end{eqnarray*}
			for each $j = 1,\ldots,l_n + 1$. Note that
			\begin{eqnarray*}
				\big|\sum_{i = 1}^n \big\{ X_{i,1} X_{i+l_n,2} - E(X_{i,1} X_{i+l_n,2})\big\} \big|
				\le (l_n + 1)\sup_{j \le l_n + 1} \Big |\sum_{i = 1}^{n/(l_n+1)} \big\{Y_{ji} - E(Y_{ji})\big\}\Big|.
			\end{eqnarray*}
			Therefore,
			\begin{eqnarray*}
				&& \pr\Big[\Big|\sum_{i = 1}^n \big\{ X_{i,1} X_{i+l_n,2} - E(X_{i,1} X_{i+l_n,2})\big\} \Big|
				> 3 (l_n+1) x \Big] \\
				& &\le (l_n + 1) \sup_{j \le l_n+1}\pr\Big[ \Big |\sum_{i = 1}^{n/(l_n+1)} \big\{Y_{ji} - E(Y_{ji})\big\}\Big| > 3 x\Big].
			\end{eqnarray*}
			The lemma is proved.
		\end{proof}
		
		Lemmas 5, 6 and 7 below are based on the autoregressive model with order 1 under $\|\bA_1\|_2 \le \delta < 1$. Similar techniques can be applied to the general cases of order $d$. 
		For  $j, k = 1,\ldots,p$, define $\wh{\Sigma}_{jk} = n^{-1}\sum_{t=1}^n y_{j,t}y_{k,t}$ and $\Sigma_{jk} = E(\hat{\Sigma}_{jk})$. For $i = 1,\ldots,p,$ let $\bbe_{(i)} = (\ep_{i,2},\ldots,\ep_{i,n})^{\T}$ and $\bx_{(i)} = (y_{i,1},\ldots,y_{i,n-1})^{\T}$. We should note that Lemmas 5 and 6 have the same rate expressions but the actual rates are different, since they are under Conditions 4(i) and 4(ii), respectively.
		
		\begin{lemma}
			\label{lem5}
			Suppose that Conditions (1)--(3) and 4(i) in Section 3.1 of the original article hold.
			\begin{description}
				\item (i) For $j,k = 1,\ldots, p$, there exist positive constants $C_1, C_2$ and $C_3$ free of $(j,k,n,p)$ such that
				$$
				\pr\left( \left|\wh{\Sigma}_{jk} - \Sigma_{jk}\right| > x\right) \le {C_1 n \over (n x)^q} + C_2 \exp\big(- C_3 n x^2\big)
				$$
				holds for $ x > 0$; consequently, this leads to the following uniform convergence rate:
				$$
				\sup_{1 \le j,k \le p}\left|\wh{\Sigma}_{jk} - \Sigma_{jk}\right| = O_P\Big\{(n^{-1}\log p)^{1/2}\Big\}.
				$$
				
				\item (ii) For $j,k = 1,\ldots, p$, there exist positive constants $C_1, C_2$ and $C_3$ free of $(j,k,n,p)$ such that
				$$
				\pr\big\{|\bbe_{(j)}^{\T}\bx_{(k)}| \ge x \big\} \le {C_1 n \over x^{2 q}} + C_2 \exp\big(- C_3 x^2\big)
				$$
				holds for $ x > 0$; in particular, we have
				$$
				\sup_{1 \le j,k \le p}|\bbe_{(j)}^{\T}\bx_{(k)}| = O_P\Big\{(n\log p)^{1/2}\Big\}.
				$$
			\end{description}
		\end{lemma}
		
		\begin{proof}
			Here we prove part (i) only. Part (ii) can be proved analogously. Let $\mu_{q} = \sup_{j\le p} \|\varepsilon_{j0}\|_{q}$ for $q \ge 2$. To use the results of Lemma 1, we just need to bound the physical dependent measure of $y_{j,t}y_{k,t}$ for each $j$ and $k$, denoted by $ \tilde{\theta}_{i,q,j,k} = \|y_{j,i}y_{k,i} - y_{j,i}^*y_{k,i}^*\|_{q}$ with $y_{j,i}^*$ being the coupled process of $y_{j,i}$. Denote the physical dependent measure of $y_{j,i}$ by $\theta_{i,2q,j} = \|y_{j,i} - y_{j,i}^*\|_{2q}$
			with $y_{j,i}^*$ being the coupled process of $y_{j,i}$.
			
			We will show (a) $\sup_{j \le p} \|y_{j,i}\|_{2q} \le C \mu_{2q}$; (b)
			$
			\sup_{j \le p}\theta_{i,2q,j} \le C \mu_{2q} (i+1) \delta^{i},
			$
			where $C$ is some positive constant and depends only on the spectral norm of $\bA_1$ rather than $q$.
			Observe that $\|y_{j,i}y_{k,i} - y_{j,i}^*y_{k,i}^*\|_{q} \le \|y_{j,i}y_{k,i} - y_{j,i}^*y_{k,i}\|_{q} + \|y_{j,i}y_{k,i} - y_{j,i}y_{k,i}\|_{q}$ and hence
			\begin{eqnarray*}
				\|y_{j,i}y_{k,i} - y_{j,i}^*y_{k,i}^*\|_{q} \le  \sup_{j \le p} \|y_{j,i}\|_{2q} \big(\theta_{i,2q,j} + \theta_{i,2q,k} \big).
			\end{eqnarray*}
			If both bounds (a) and (b) are obtained, then,
			$$
			\tilde{\Theta}_{m,q} = \sup_{j,k \le p } \sum_{i= m}^{\infty} \tilde{\theta}_{i,2 q,j,k} \le C \mu_{2q}^2 \sum_{i = m}^\infty (i+1) \delta^i
			\le C \mu_{2q}^2(1-\delta)^{-2} (m+1) \delta^m = o(m^{-\alpha})
			$$
			for any $\alpha > 1 $.  Applying Lemma 1 we prove part (i).
			
			Let us turn to bound $\sup_{j \le p} \|y_{j,i}\|_{2q}$. Let $\bA_1^l$ be $(a_{l,jk})_{j,k \le p}$ with $l \ge 1$. Since $\bA_1^l$ is a banded matrix with the bandwidth $\min(2l k_0 + 1,p)$, we can bound $\|\bA_1^l\|_{\infty}$ by
			\begin{eqnarray}
			\label{Ainfty}
			\|\bA_1^l\|_{\infty} = \max_{j \le p}\sum_{k=1}^p |a_{l,jk}|\le \{\min(2lk_0+1,p)\}^{1/2} \|\bA_1^l\|_2 \le C(2lk_0+1) \delta^l, l \ge 1,
			\end{eqnarray}
			which implies that $\|\bA_1^l\|_{\infty} \le C(2k_0+1)(l+1) \delta^l, l \ge 0$.
			Using the innovation representation $\by_t = \sum_{l = 0}^{\infty} \bA_1^l \bep_{t-l}$,
			we get
			$$
			\|y_{j,i}\|_{2q} \le \sum_{l = 0}^\infty \|\sum_{k=1}^p a_{l,jk} \varepsilon_{k,i-l}\|_{2q} \le \sum_{l = 0}^\infty \sum_{k=1}^p |a_{l,jk}| \|\varepsilon_{k,i-l}\|_{2q}.
			$$
			As a result,
			$
			\sup_{j \le p}\|y_{j,i}\|_{2q} \le C(2k_0+1)\mu_{2q}\sum_{l = 0}^\infty (l+1) \delta^l = C(2k_0+1) (1- \delta)^{-2} \mu_{2q}  < \infty.
			$
			Similarly, we can bound $\sup_{j\le p}\theta_{i,2q,j}$ above by $C  (i+1) \delta^i$ with some positive constant $C$ since we have a nice inequality
			$$
			\|y_{j,i} - y_{j,i}^*\|_{2q} = \|\sum_{k=1}^p a_{i,j k} \big(\varepsilon_{k,0}-\varepsilon_{k,0}^*\big)\|_{2q} \le 2 \mu_{2q} \|\bA_1^i\|_{\infty}.
			$$
			The proof is complete.
		\end{proof}
		
		\begin{lemma}
			\label{lem6}
			Suppose that Conditions (1)--(3) and 4(ii)  in Section 3.1 of the original article hold. Then we have \\
			
			(i)  $
			\underset{1 \le j,k \le p}{\sup}|\wh{\Sigma}_{jk} - \Sigma_{jk}| = O_P\Big\{(n^{-1}\log p)^{1/2} \Big\};
			$
			(ii)  
			$
			\underset{1 \le j,k \le p}{\sup}|\bbe_{(j)}^{\T}\bx_{(k)}| = O_P\Big\{ (n\log p)^{1/2}\Big\}.
			$
		\end{lemma}
		\begin{proof}
			Here we prove part (i) only. The proof of part (ii) can be derived similarly.
			
			Note that $\by_t = \bA_1 \by_{t-1} +  \bep_{t-l}$ and $\|\bA_1\|_2 \le \delta < 1$. Let $\bA_1^l$ be $(a_{l,jk})_{j,k \le p} $. For each $j$, $y_{j,t} = \sum_{l = 0}^\infty \sum_{m=1}^p a_{l,jm} \varepsilon_{m,t-l}$ converges almost surely. Write $\eta_{j,lt} = \sum_{m=1}^p a_{l,jm} \varepsilon_{m,t-l}$ for $l \ge 0$. We divide $y_{j,t}$ into two terms $y_{jt} = \sum_{l=0}^{N_{n}} \eta_{j,lt} + \sum_{l = N_{n} + 1}^\infty \eta_{j,lt}$. Here we choose $N_n$ to be $ N_\delta \log (n)$ with $N_\delta > (1+ \alpha)\alpha^{-1} ( -\log \delta)^{-1}$. Hence, $n\wh{\Sigma}_{jk}$ can be expressed as
			\begin{eqnarray*}
				n\wh{\Sigma}_{j k} & = & \sum_{l,l' = 0}^{N_n} \left(\sum_{t = 1}^n \eta_{j,l t} \eta_{k,l't}\right) + \sum_{l,l' = N_n + 1}^{\infty} \left(\sum_{t = 1}^n \eta_{j,l t} \eta_{k,l't}\right) \\
				&& + \sum_{l = 0}^{N_n} \sum_{l' = N_{n} + 1}^{\infty}\left(\sum_{t = 1}^n \eta_{j,l t} \eta_{k,l't}\right) + \sum_{l = N_n + 1}^{\infty} \sum_{l' = 0}^{N_{n}} \left(\sum_{t = 1}^n \eta_{j,l t} \eta_{k,l't}\right) \\
				& = & S_{jk,1} + S_{jk,2} + S_{jk,3} + S_{jk,4},
			\end{eqnarray*}
			and
			$
			n\Big(\wh{\Sigma}_{j k} - \Sigma_{j k}\Big) = \sum_{m = 1}^4 \big\{S_{jk,m} - E(S_{jk,m})\big\}.
			$
			Let us handle the first term $S_{jk,1} - E(S_{jk,1})$. Note that if $\sup_{m,l}E\{\exp(|t_0\varepsilon_{m,l}|^{2\alpha})\} < \infty$,
			$$
			\zeta_{\varepsilon} = \sup_{m,l,m',l'} E\big\{\exp\big(t_0|\varepsilon_{m,l} \varepsilon_{m',l'}|^{\alpha}\big)\big\} <\infty.
			$$
			By Lemma 4, we obtain the following equality,
			\begin{eqnarray*}
				&& \pr\Big[\Big|\sum_{t = 1}^n \big\{\varepsilon_{m,t-l} \varepsilon_{m',t-l'}
				- E(\varepsilon_{m,t-l} \varepsilon_{m',t-l'})\big\}\Big| > 3 (l_n + 1) x\Big] \\
				& &\le   (l_n+1) C_1\exp\left( - { x^2 \over C_2 n  + C_3 n^{1-\alpha \over 2-\alpha} x}\right) + C_1 (l_n + 1) \exp\left( - { x^{2\alpha} \over C_2 n^{\alpha \over 2-\alpha} + C_3 x^{\alpha}}\right)\\
				& &   + C_1 (l_n + 1) n \exp\left(-C_4 x^{\alpha}\right)
			\end{eqnarray*}
			for some positive constants $C_j (j= 1,\ldots,4)$, where $l_n = |l -l'|$. Taking $x = C (n \log p)^{1/2}$ for some large constant $C>0$, we
			derive 
			the following term
			\begin{eqnarray*}
				\tilde{\eta}_n = \sup_{m, m'\le p, l,l' \le N_{n}} (l' + 1)^{-2} (l + 1)^{-2} \left|\sum_{t = 1}^n \big\{\varepsilon_{m,t-l} \varepsilon_{m',t-l'} - E(\varepsilon_{m,t-l} \varepsilon_{m',t-l'})\big\}\right|
			\end{eqnarray*}
			with the convergence rate $ O_P\big\{(n \log p)^{1/2}\big\}$. Observe that
			\begin{eqnarray}
			\label{lemma6_1}
			\Big|\sum_{t = 1}^n \big\{\eta_{j,l t} \eta_{k,l't} - E(\eta_{j,l t} \eta_{k,l't})\big\}\Big|
			&&\le C(2 k_0+1)^2(l+1)^3(l'+1)^3 \delta^{l + l'} \tilde{\eta}_n,
			\end{eqnarray}
			and  $\sum_{l = 0}^{N_\delta} (l+1)^3 \delta^l < \infty$. Therefore,
			$$
			\sup_{j,k \le p}\Big|S_{jk,1} - E(S_{jk,1}) \Big| \le C {(2k_0 + 1)^2 } \tilde{\eta}_n =  O_P\big\{(n \log p)^{1/2}\big\}.
			$$
			
			Consider the second term. Since $\sup_{m,l}E\big\{\exp(|t_0\varepsilon_{m,l}|^{2\alpha})\big\} < \infty$,
			$
			\tilde{\zeta}_{q,\varepsilon} = \sup_{m,l,m',l'} \Big\|\varepsilon_{m,l} \varepsilon_{m',l'}\Big\|_q \le C q^{1/\alpha}
			$ for any $q > 2$. Now we bound $\Big\|S_{jk,2} - E S_{jk,2}  \Big\|_q.$ To be specific,
			\begin{eqnarray}
			\label{lemma6_2}
			\Big\|\sum_{t = 1}^n \big\{\eta_{j,l t} \eta_{k,l't} - E(\eta_{j,l t} \eta_{k,l't})\big\} \Big\|_q &\le & n\sum_{m,m' =1}^p |a_{l,j m}| |a_{l,km'}| \sup_{m,m',l,l'}\Big\|\varepsilon_{m,l} \varepsilon_{m,l'} \Big\|_q \nonumber\\
			&\le& n(2k_0+1)^2(l+1)(l'+1)\delta^{l+l'} \tilde{\zeta}_{q,\varepsilon}.
			\end{eqnarray}
			Hence,
			\begin{eqnarray*}
				\Big\| S_{jk,2} - E (S_{jk,2})  \Big\|_q  & \le & C n  q^{1/\alpha} \sum_{l,l' = N_{n} + 1}^\infty (l+1)(l'+1)\delta^{l+l'} \le  C \cdot n N_{n}^2 \delta^{2 N_n} q^{1/\alpha}.
			\end{eqnarray*}
			Write $\eta_{n2} = \big(n N_{n}^2 \delta^{2 N_n}\big)^{-1}\big\{S_{jk,2} - E (S_{jk,2})\big\}$. It follows from Lemma 2 that there exists a constant $\lambda > 0$ such that
			$
			E\big\{\exp(\lambda|\eta_{n2}|^{\alpha})\big\} < \infty.
			$
			Consequently, for a large constant $C > 0$, we have that
			\begin{eqnarray*}
				&&\pr\Big\{ \sup_{j,k \le p}\Big| S_{jk,2} - E (S_{jk,2})  \Big| > C (\log n)^2\Big\}\\
				&&\le  O(1) p^2\exp\Big\{- t C^\alpha \cdot n (\log n)^{-2\alpha} (\log n)^{2\alpha}\Big\} \to 0,
			\end{eqnarray*}
			as $n \to \infty$, which implies that
			$
			\sup_{j,k \le p}\Big| S_{jk,2} - E (S_{jk,2})  \Big| = O_P\Big\{ (\log n)^2 \Big\} = o_P\big\{(n\log p)^{1/2}\big\}.
			$
			Similarly, we can prove that
			$
			\sup_{j,k \le p}\Big| S_{jk,m} - E(S_{jk,m})  \Big| = o_P\big\{(n\log p)^{1/2}\big\}, m = 3,4.
			$
			
			Finally putting together the convergence rate results for the four terms we conclude
			$$
			\sup_{j,k \le p}\Big|\wh{\Sigma}_{jk} - \Sigma_{jk}\Big| = O_P\big\{(n^{-1}\log p)^{1/2}\big\}.
			$$
			The proof is complete.
		\end{proof}
		
		\begin{lemma}
			\label{lem7}
			Suppose that Conditions (1)--(3) and 4(i) or 4(ii) in Section 3.1 of the original article hold. Then, for each finite $k$ with $ k \ge k_0$,
			$$
			\sup_{1 \le i\le p}\Big|{\rss_i(k) \over n\sigma_i^2} - 1\Big| = O_P\Big\{(n^{-1}\log p)^{1/2}\Big\},
			$$
			as $n\to \infty$, where $\rss_i(k)$ is defined in (2.6) and $\sigma_i^2$ is the $(i,i)$-th element of $\bSigma_{\bep}$.
		\end{lemma}
		\begin{proof}
			For $k > k_0$, the term $\rss_i(k)$ can be decomposed as
			$$
			\rss_i(k) = \bbe_i^{\T}\bbe_i - \bbe_i^{\T}\bX_{i}\big(\bX_{i}^{\T}\bX_{i}\big)^{-1}\bX_{i}^{\T}\bbe_i =
			R_{i1} - R_{i2},
			$$
			where $\bbe_{(i)} = (\ep_{i,2},\ldots,\ep_{i,n})^{T}$, and $\bX_{i}$ is a $(n-1) \times \tau_i(k)$ matrix with $\bx_{i,1+j}$ as its $j$-th row.
			We will show below that, under Assumptions (1)--(3) and 4(i) or 4(ii),
			$$
			\mbox{(a)} \quad \sup_{i \le p} \Big|R_{i1} - n \sigma_i^2\Big| = O_P\{ (n \log p)^{1/2}\};
			~~ \mbox{(b)} \quad \sup_{i \le p} |R_{i2}| = O_P\big(\log p\big).
			$$
			With results in (a) and (b), it follows that
			$$
			\sup_{i \le p}\Big|{\rss_i(k) \over n \sigma_i^2} - 1\Big| \le
			\sup_{i \le p}\Big|{R_{i 1} \over n \sigma_i^2}-1\Big| +
			\sup_{i \le p}\Big|{R_{i2} \over n \sigma_i^2} \Big| = O_P\big\{ (n^{-1}\log p)^{1/2}\big\}.
			$$
			
			Suppose first that Condition 4(i) holds.  Consider the term $R_{i1} - n \sigma_i^2$. Lemma 1 shows that
			$$
			\sup_{i \le p} \Big|\bbe_i^{\T}\bbe_i - n \sigma_i^2\Big| = O_P\big\{ (n \log p)^{1/2}\big\}.
			$$
			Let us handle the term $\sup_{i\le p}|R_{i2}|$. Define
			$$
			\mathcal{A}_n = \Big\{ \inf_{i \le p}\lambda_{\min}\big(n^{-1}\bX_{i}^{\T}\bX_{i}\big) > \kappa_1 (1-\eta)\Big\}
			$$
			with $0 < \eta < 1$. It follows from Lemma 5(i) and Condition 3 that
			$
			P(\mathcal{A}_n) \to 1
			$
			as $n \to \infty$. On the event $\mathcal{A}_n$, the term $\sup_{i\le p}|R_{i2}|$ can be bounded above by
			$
			\big( \kappa_1 (1- \eta)\big)^{-1} k_0\sup_{j,k \le p} n^{-1}|\bbe_{(j)}^{\T} \bx_{(k)}|^2.
			$
			By Lemma 5 (ii), we obtain that,
			$
			\sup_{j,k \le p} |\bbe_{(j)}^{\T} \bx_{(k)}| =O_P\big\{ (n \log p)^{1/2}\big\},
			$
			which implies that (b) holds.
			
			Suppose that Condition 4(ii) holds.  Consider the term $R_{i1} - n \sigma_i^2$. By Lemma 3 and taking $x = C (n \log p)^{1/2}$ with large constant $ C > 0$, we have that
			$$
			\sup_{i \le p} \Big|\bbe_i^{\T}\bbe_i - n \sigma_i^2\Big| = O_P\big\{(n \log p)^{1/2}\big\}.
			$$
			Similarly, we can establish (b) from Lemma 6.  The proof is complete.
		\end{proof}

		\section{An additional Simulation Study}
		
		We conduct an additional Monte Carlo experiment to examine the proposed
		methodology.
		We consider a banded vector autoregressive model with the sample size $n+2$, where the last
		two observations are used to calculate the one-step and two-step ahead
		post-sample prediction errors.
		
		The data  were generated from a vector autoregressive model with $d = 1$ and
		the banded coefficient matrix  $A$ specified in scenario (1) in the
		paper. We set $n =200$, $p =
		100, 200$, $k_0 = 2$, and each setting was repeated  $100$ times. To
		mimic the real world with the true ordering unknown, we considered three
		other orderings through random permutation. The first ordering was
		generated through local permutation, where we partitioned the components
		of $y_t$  into $[p/5]$ groups with each group containing 5 components.  We
		then performed a random permutation within each group.
		The other two orderings were
		generated through permutating the whole components of $y_t$ together.
		Also included in the comparison is the sparse autoregressive model determined by lasso.
		Table \ref{table7} below reports simulation results of Bayesian information criterion scores and prediction errors.
		It indicates that the model with the true ordering offers the best post-sample
		prediction, followed by the model with the local permutation only, and then
		the lasso-based model, while the two models with arbitrary permutations perform
		the worst.
		
		\newpage
		
		\begin{table}
			\caption{Average Bayeysian information criterion values, estimated bandwidth parameter and  one-step-ahead and two-step-ahead post-sample predictive errors over $100$ replications, 	
				with their corresponding standard errors in parentheses.}
			\begin{center}
				\begin{tabular}{l|c|c|c|c}
					\hline \hline
					Ordering &  BIC & Bandwidth  & One-step ahead & Two-step ahead\\
					\hline \hline
					\multicolumn{4}{c}{Case 1: $n= 200, p = 100$}  \\
					\hline
					True ordering       & 546(3.5) & 1.78(0.52) & 0.787(0.06) & 0.837(0.07) \\
					Local permutation   & 549(5.6) & 2.14(0.83) & 0.788(0.06) & 0.837(0.07) \\ 		
					Random permutation  & 546(11)  & 0.71(0.90) & 0.828(0.07) & 0.848(0.08) \\ 	
					Random permutation  & 547(13)  & 0.70(1.06) & 0.827(0.06) & 0.848(0.08) \\
					\hline
					Lasso &	--& --& 0.823(0.06) & 0.846(0.08) \\		
					\hline \hline					
					\multicolumn{4}{c}{Case 2: $n= 200, p = 200$}  \\
					\hline
					True ordering       & 1093(4.8) & 1.87(0.36) & 0.786(0.04) & 0.830(0.05) \\
					Local permutation   & 1102(10)  & 2.39(0.75) & 0.787(0.04) & 0.829(0.05) \\		
					Random permutation  & 1098(22)  & 0.89(0.80) & 0.829(0.05) & 0.839(0.05) \\		
					Random permutation  & 1096(20)  & 0.78(0.70) & 0.831(0.05) & 0.838(0.05) \\		
					\hline
					Lasso &	-- & -- & 0.827(0.05) & 0.838(0.05)	\\			
					\hline \hline					
				\end{tabular}
			\label{table7}
			\end{center}
		\end{table}

\end{document}